\documentclass[11pt]{article}
\usepackage{graphicx} % Required for inserting images

\title{Courcelle's Theorem for Lipschitz Continuity}

\author{
 Tatsuya Gima\thanks{Hokkaido University. Supported by JSPS KAKENHI Grant Numbers JP24K23847, JP25K03077.}\and
 Soh Kumabe\thanks{CyberAgent.}\and
 Yuichi Yoshida\thanks{National Institute of Informatics. Supported by JSPS KAKENHI Grant Number JP24K02903.}}
\date{}

\usepackage{fullpage}
\usepackage{amsmath,amsthm,amssymb}

\usepackage{mathtools}

\usepackage{xspace}
\usepackage{todonotes}

\DeclarePairedDelimiter{\angles}{\langle}{\rangle}
\DeclarePairedDelimiter{\braces}{\{}{\}}
\DeclarePairedDelimiter{\parens}{(}{)}

\newcommand{\FunctionNameFont}[1]{\mathsf{#1}}
\newcommand{\ConstantNameFont}[1]{\mathtt{#1}}
\newcommand{\Cons}[1]{\mathtt{#1}}
\newcommand{\ClassNameFont}[1]{\textsf{#1}}

\newcommand{\CMSO}[2]{$\ClassNameFont{C}_{#1}\ClassNameFont{MSO}^{#2}$\xspace}
\newcommand{\MSO}{$\ClassNameFont{MSO}$\xspace}
\newcommand{\MSOo}{$\ClassNameFont{MSO}_1$\xspace}
\newcommand{\MSOt}{$\ClassNameFont{MSO}_2$\xspace}

\newcommand{\dunion}{\uplus} %% disjoint union of sets
\newcommand{\HR}{$\ClassNameFont{HR}$\xspace}

\newcommand{\src}{\FunctionNameFont{src}} %% sources
\newcommand{\paracom}{\mathbin{/\hspace{-1mm}/}}
\newcommand{\fg}{\FunctionNameFont{fg}}
\newcommand{\kgraph}[1]{{#1}-graph}
\newcommand{\kgraphs}[1]{{#1}-graphs} %% plural
\newcommand{\sat}{\FunctionNameFont{sat}}

\newcommand{\height}{\FunctionNameFont{height}}

\newcommand{\nil}{\ConstantNameFont{nil}} %% plural

\newcommand{\adj}{\FunctionNameFont{adj}}
\newcommand{\true}{\ConstantNameFont{True}}
\newcommand{\false}{\ConstantNameFont{False}}

\newcommand{\Card}{\FunctionNameFont{Card}}

\newcommand{\nosrc}{\FunctionNameFont{nosrc}}

\newcommand{\DP}{\FunctionNameFont{DP}}
\newcommand{\softmax}{\FunctionNameFont{softmax}}
\newcommand{\softmin}{\FunctionNameFont{softmin}}
\newcommand{\softargmax}{\FunctionNameFont{softargmax}}
\newcommand{\softargmin}{\FunctionNameFont{softargmin}}
\newcommand{\OPT}{\FunctionNameFont{opt}}
\newcommand{\E}{\mathbb{E}}
\newcommand{\EM}{\FunctionNameFont{EM}}
\newcommand{\TV}{\FunctionNameFont{TV}}

\newcommand{\cw}{\mathrm{cw}}
\newcommand{\tw}{\mathrm{tw}}

\usepackage{algpseudocode}
\usepackage{bm}
\usepackage[linesnumbered,ruled,vlined]{algorithm2e}
\SetKwProg{Procedure}{Procedure}{}{}

\allowdisplaybreaks

\newcommand{\apxmark}{\star}

\usepackage[colorlinks=true]{hyperref}
\usepackage[capitalize]{cleveref}

\newtheorem{theorem}{Theorem}
\newtheorem{proposition}[theorem]{Proposition}
\newtheorem{lemma}[theorem]{Lemma}
\newtheorem{observation}[theorem]{Observation}
\newtheorem{corollary}[theorem]{Corollary}

\theoremstyle{definition}
\newtheorem{definition}[theorem]{Definition}

\begin{document}

\maketitle
\begin{abstract}
    Lipschitz continuity of algorithms, introduced by Kumabe and Yoshida (FOCS'23), measures the stability of an algorithm against small input perturbations. 
    Algorithms with small Lipschitz continuity are desirable, as they ensure reliable decision-making and reproducible scientific research.
    Several studies have proposed Lipschitz continuous algorithms for various combinatorial optimization problems, but these algorithms are problem-specific, requiring a separate design for each problem.
    
    To address this issue, we provide the first algorithmic meta-theorem in the field of Lipschitz continuous algorithms.
    Our result can be seen as a Lipschitz continuous analogue of Courcelle's theorem, which offers Lipschitz continuous algorithms for problems on bounded-treewidth graphs. 
    Specifically, we consider the problem of finding a vertex set in a graph that maximizes or minimizes the total weight, subject to constraints expressed in monadic second-order logic (\MSOt). We show that for any $\varepsilon>0$, there exists a $(1\pm \varepsilon)$-approximation algorithm for the problem with a polylogarithmic Lipschitz constant on bounded treewidth graphs.
    On such graphs, our result outperforms most existing Lipschitz continuous algorithms in terms of approximability and/or Lipschitz continuity.
    Further, we provide similar results for problems on bounded-clique-width graphs subject to constraints expressed in \MSOo.
    Additionally, we construct a Lipschitz continuous version of Baker's decomposition using our meta-theorem as a subroutine.
\end{abstract}

\section{Introduction}

Lipschitz continuity of algorithms, introduced by Kumabe and Yoshida~\cite{KumabeY23}, is a measure of an algorithm's stability in response to errors or small perturbations in the input for weighted optimization problems. 
Roughly speaking, it is the maximum ratio of the (weighted) Hamming distance between the outputs of an algorithm for two different weights to the $\ell_1$ distance between those weights (see \Cref{sec:prelim:lipschitz} for the precise definition).
It is desirable for algorithms to have small Lipschitz constants, as large constants can undermine the reliability of decision-making and the reproducibility of scientific research.
%Using algorithms with large or unbounded Lipschitz constants
% , meaning that the output is significantly affected by noise or errors, 
%might undermine 
%\ynote{maybe we need more comments on why we need continuity.}\snote{wrote something}
%Since exact \ynote{even deterministic} algorithms do not exhibit Lipschitz continuity~\cite{KumabeY23}, these algorithms are approximation algorithms. \ynote{we can move this sentence to Sec 1.3.}

Since its introduction, Lipschitz continuous algorithms have been proposed for various optimization problems~\cite{KumabeY23,KumabeY25}.
However, we need to design a different algorithm and analyze its Lipschitz continuity for each problem, which can be impractical.
To address this limitation, we present an algorithmic meta-theorem for Lipschitz continuous algorithms, which can be seen as Lipschitz continuity analogue of celebrated Courcelle's theorem~\cite{Courcelle90}.

To present our results, we first introduce some notation.
Let $G=(V,E)$ be a graph of $n$ vertices with treewidth $\tw$, $\varphi(X)$ be an \MSOt formula with a free vertex set variable $X$ (see \Cref{sec:prelim_logic} for details about logic), and $w\in \mathbb{R}_{\geq 0}^{V}$ be a weight vector. 
We consider the problem of finding a vertex subset $X\subseteq V$ such that $G\models \varphi(X)$ and $w(X)$ is maximized, which we call the \emph{\MSOt maximization problem}.
We also consider the problem of finding $X$ such that $w(X)$ is minimized, which we call the \emph{\MSOt minimization problem}.
% \snote{letter to represent treewidth: is $k$ better?}
% The main result of this paper is to provide algorithmic meta-theorems that yield Lipschitz continuous algorithms.
Our main results are the following.
\begin{theorem}\label{thm:alg_max}
    For any $\varepsilon\in (0,1]$, there is a $(1-\varepsilon)$-approximation algorithm for the \MSOt maximization problem with Lipschitz constant $O\left((f(\tw,|\varphi|)+\log \varepsilon^{-1} + \log \log n)\varepsilon^{-1}\log^2 n\right)$, where $f$ is some computable function. The time complexity is bounded by $O\left(f(\tw,|\varphi|)n\right)$.
% \ynote{is there a strong reason not to use $n$ to denote the number of vertices?}\snote{replaced by $n$}
\end{theorem}
\begin{theorem}\label{thm:alg_min}
    For any $\varepsilon\in (0,1]$, there is a $(1+\varepsilon)$-approximation algorithm for the \MSOt minimization problem with Lipschitz constant $O\left((f(\tw,|\varphi|)+\log \varepsilon^{-1} + \log \log n)\varepsilon^{-1}\log^2 n\right)$, where $f$ is some computable function. The time complexity is bounded by $O\left(f(\tw,|\varphi|)n\right)$.
\end{theorem}
We note that the trivial upper bound on Lipschitz constant is $n$~\cite{KumabeY25}; therefore the bounds in the theorems above are significantly smaller for fixed $\tw$ and $\varphi$.
We also remark that our meta-theorems yield randomized approximation algorithms. 
This is necessary since, for most problems, it is known that exact or deterministic algorithms cannot be Lipschitz continuous~\cite{KumabeY23}.

When the treewidth of the input graph is bounded by a constant, \Cref{thm:alg_max,thm:alg_min} provide algorithms that outperform existing algorithms in terms of approximability and/or Lipschitz continuity:
\begin{itemize}  
\itemsep=0pt
\item For the minimum weight vertex cover problem, \Cref{thm:alg_min} yields an algorithm with a better approximation ratio than the previous $2$-approximation algorithm with Lipschitz constant $4$~\cite{KumabeY25}, at the polylogarithmic sacrifice of the Lipschitz constant.
\item For the minimum weight feedback vertex set problem, \Cref{thm:alg_min} outperforms the previous $O(\log n)$-approximation algorithm with Lipschitz constant $O(\sqrt{n}\log^{3/2}n)$~\cite{KumabeY25} in terms of both approximability and Lipschitz continuity.
\item For the maximum weight matching problem\footnote{Although the output of this problem is an edge set rather than a vertex set, this problem can be expressed by an \MSOt formula on a new graph $G'=(V\cup E, \{(v,e)\colon e \text{ is incident to }v\})$}, \Cref{thm:alg_max} yields an algorithm with a better approximation ratio than the previous $\left(\frac{1}{8}-\varepsilon\right)$-approximation algorithm with Lipschitz constant $O(\varepsilon^{-1})$~\cite{KumabeY23}, at the polylogarithmic sacrifice of the Lipschitz constant.
\item For the shortest path problem, \Cref{thm:alg_min} slightly improves the Lipschitz continuity compared to the previous $(1-\varepsilon)$-approximation algorithm with Lipschitz constant $O(\varepsilon^{-1}\log^3 n)$~\cite{KumabeY23}, without losing approximability.
\end{itemize}

For a fixed $\varphi$, by explicitly specifying the transitions in the dynamic programming within the algorithm, we can provide more precise bounds on the function $f$ that appears in \Cref{thm:alg_max,thm:alg_min}. 
In particular, considering the case that $\varphi$ is an \MSOt formula representing the independent set constraint, we have the following.
\begin{theorem}\label{thm:alg_independentset}
    For any $\varepsilon\in (0,1]$, there is a $(1-\varepsilon)$-approximation algorithm for the maximum weight independent set problem with Lipschitz constant $O\left((\tw+\log \varepsilon^{-1} + \log \log n)\varepsilon^{-1}\log^2 n\right)$. The time complexity is bounded by $2^{O(\tw)}n$.
\end{theorem}
This result is surprising as the dependence of $\tw$ on the Lipschitz constant is subexponential and, even more remarkably, linear.
We note that through similar arguments, analogous results hold for several other problems, such as the minimum weight vertex cover problem and the minimum weight dominating set problem.

%For the set cover problem\snote{This problem can also be expressed by an \MSOt formula on graphs with a free vertex set variable.}, four algorithms with different approximability and Lipschitzness guarantees are previously known~\cite{KumabeY25}. \snote{kakikake}
% \snote{need set cover? it can be expressed by \MSOt on some graph, but I am not sure it is interesting and worth noting.}

% \ynote{we should compare these results with known upper bounds like vertex cover, matching, feedback vertex set.}

We further demonstrate that \Cref{thm:alg_max,thm:alg_min} lead to Lipschitz continuous version of Baker's technique~\cite{Baker94}.
As a representative example, we consider the maximum weight independent set problem on planar graphs, where a vertex subset $X$ is an \emph{independent set} of a graph $G$ if no two vertices in $X$ are adjacent in $G$.
We prove the following.
\begin{theorem}\label{thm:baker}
    For any $\varepsilon\in (0,1]$, there is a $(1-\varepsilon)$-approximation algorithm for the maximum weight independent set problem on planar graphs with Lipschitz constant $O\left((\varepsilon^{-1} + \log \log n)\varepsilon^{-1}\log^2 n\right)$. The time complexity is bounded by $2^{O(\varepsilon^{-1})}n$.
\end{theorem}
Using similar algorithms, Lipschitz continuous PTASes can be obtained for many problems, including the minimum weight vertex cover problem and the minimum weight dominating set problem.

As a lower bound result, we prove that the Lipschitz constant of any $(1-\varepsilon)$-approximation algorithm for some \MSOt maximization problem is large on general graphs, which justifies considering the \MSOt maximization problem on a restricted graph class. 
Specifically, we consider the \emph{max ones problem}~\cite{creignou2001complexity,khanna2001approximability}, where the instance is a 3CNF formula over a variable set $X$ and a weight function $w:X \to \mathbb R_{\geq 0}$, and the goal is to find a satisfying assignment $\sigma:X \to \{0,1\}$ that maximizes the weight $\sum_{x \in \sigma^{-1}(1)}w(x)$.
It is easy to see that there exists a fixed \MSOt formula $\varphi(X)$ such that the max ones problem can be reduced to a problem on bipartite graphs, where the task is to find a vertex set $X$ that maximizes the weight $\sum_{x\in X} w(x)$ subject to $\varphi$.
We show the following.
\begin{theorem}\label{thm:lowerbound}
    There exist $\varepsilon,\delta>0$ such that any $(1-\varepsilon)$-approximation algorithm for the max ones problem has Lipschitz constant $\Omega(n^{\delta})$, where $n$ is the number of variables.
%    \ynote{we can show lower bound for a variant of 3SAT where we want to find a satisfying assignment that maximizes the weight of ones (a variant of min-one SAT). I'm not sure it's interesting.}
\end{theorem}
We note that it is possible to construct an algorithm for the max ones problem with a Lipschitz constant similar to that in \Cref{thm:alg_independentset}, where $\tw$ denotes the treewidth of the ``incidence graph'' of the 3CNF instance (see \Cref{sec:dp_maxone} for details). 
Since $\tw$ is upper bounded by $n$ for any instance, \Cref{thm:lowerbound} implies that the dependency of the Lipschitz constant on $\tw$ cannot be improved to $\tw^{o(1)}$.
% Since the treewidth of a graph is at most $n-1$, \Cref{thm:lowerbound} implies that the Lipschitz constant bound of \Cref{thm:alg_independentset} cannot be improved to $k^{o(1)}$.

% \ynote{we need to mention how we proved the theorem and the technical challenges encountered in the proof.}

% \snote{added the following}

We also prove analogous meta-theorems when parameterizing by clique-width.  
Let an \emph{\MSOo maximization (resp., minimization) problem} be the variant of an \MSOt problem in which the formula $\varphi$ is restricted to \MSOo.  
Denote by $\cw$ the clique-width of the graph $G$.  
We then obtain the following.
\begin{theorem}\label{thm:alg_max_clique}
    For any $\varepsilon\in (0,1]$, there is a $(1-\varepsilon)$-approximation algorithm for the \MSOo maximization problem with Lipschitz constant $O\left((f(\cw,|\varphi|)+\log \varepsilon^{-1} + \log \log n)\varepsilon^{-1}\log^2 n\right)$, where $f$ is some computable function. The time complexity is bounded by $O\left(f(\cw,|\varphi|)n\right)$.
% \ynote{is there a strong reason not to use $n$ to denote the number of vertices?}\snote{replaced by $n$}
\end{theorem}
\begin{theorem}\label{thm:alg_min_clique}
    For any $\varepsilon\in (0,1]$, there is a $(1+\varepsilon)$-approximation algorithm for the \MSOo minimization problem with Lipschitz constant $O\left((f(\cw,|\varphi|)+\log \varepsilon^{-1} + \log \log n)\varepsilon^{-1}\log^2 n\right)$, where $f$ is some computable function. The time complexity is bounded by $O\left(f(\cw,|\varphi|)n\right)$.
\end{theorem}
We note that the functions $f$ in \Cref{thm:alg_max_clique,thm:alg_min_clique} are much larger than those in \Cref{thm:alg_max,thm:alg_min}. 
In particular, the bounds for \Cref{thm:alg_independentset,thm:baker} do not follow from \Cref{thm:alg_max_clique}.  

\subsection{Technical Overview}

Now we provide a technical overview of our framework.
Since the arguments for clique-width are similar, we focus here on the treewidth results and omit the clique-width case.
%only provide an overview of the results for treewidth here and do not cover the ones for clique-width.
%\snote{added this sentence}  

For simplicity, here we consider the maximum weight independent set problem on a full binary tree (a rooted tree in which every vertex has exactly 0 or 2 children).
This corresponds to the case where $\varphi(X) = \forall x \forall y ((x \in X \land y \in X) \rightarrow \lnot \adj(x, y))$. 
Let $w \in \mathbb{R}_{\geq 0}^V$ be the weight vector.

If we do not care about Lipschitzness, this problem can exactly be solved by the following algorithm.
For each vertex $v \in V$, define $\DP[v][0]$ to be an independent set in the subtree rooted at $v$ that does not include $v$ with the maximum weight. 
Similarly, define $\DP[v][1]$ to be an independent set in the subtree rooted at $v$ with the maximum weight. 
If $v$ is a leaf, then $\DP[v][0] = \emptyset$ and $\DP[v][1] = \{v\}$.
Otherwise, let $u_1$ and $u_2$ be the two children of $v$. 
We have $\DP[v][0] = \DP[u_1][1] \cup \DP[u_2][1]$, and $\DP[v][1]$ is the one with the larger weight of $X_{v,0} := \DP[u_1][0] \cup \DP[u_2][0] \cup \{v\}$ and $X_{v,1} := \DP[u_1][1] \cup \DP[u_2][1]$.
By performing this dynamic programming in a bottom-up manner, the problem can be solved exactly.

However, this algorithm is not Lipschitz continuous.
This is because when we compute $\DP[v][1]$, which of $X_{v,0}$ or $X_{v,1}$ has a larger weight is affected by small changes in the weights.
To address this issue, we use the exponential mechanism~\cite{Dwork06}. 
Specifically, for some constant $c>0$, we select $X_{v,i}$ with probability proportional to $\exp(c\cdot w(X_{v,i}))$.
This approach makes the algorithm Lipschitz continuous, with only a slight sacrifice in the approximation ratio.
Specifically, we can prove that by appropriately choosing $c$ for a given $\varepsilon' \in (0, 1]$, the increase in the Lipschitz constant can be bounded by $\tilde{O}(\varepsilon'^{-1})$ by reducing the approximation guarantee by a factor of $1 - \varepsilon'$ at each vertex where the exponential mechanism is applied.

While this approach makes the algorithm Lipschitz continuous, the Lipschitz constant is still too large when the height $h$ of the tree is large.
Specifically, to achieve an approximation guarantee of $1 - \varepsilon$, we need to set $\varepsilon' = \frac{\varepsilon}{h}$, leading to a Lipschitz constant of $h \cdot \tilde{O}(\varepsilon'^{-1}) = \tilde{O}(\varepsilon^{-1}h^2)$. 
This is larger than the trivial bound $n$ when $h = O(n)$.
We resolve this issue by using the fact that any tree has a tree decomposition of width 5 and height $O(\log n)$~\cite{BodlaenderH98}.
By performing dynamic programming with the exponential mechanism on this tree decomposition, we can obtain $(1-\varepsilon)$-approximation algorithm with Lipschitz constant $\tilde{O}(\varepsilon^{-1})$.
This argument can be naturally extended to the case where $G$ is a bounded treewidth graph.
By following the proof of Courcelle's Theorem~\cite{Courcelle90, CourcelleM93}, we further extend this argument to the case where $\varphi$ is a general \MSOt formula with a free vertex variable.

We prove \Cref{thm:lowerbound} by leveraging a lower bound on a related notion of sensitivity for the maximum cut problem~\cite{fleming2024sensitivity}, where sensitivity measures the change in the output with respect to the Hamming distance under edge deletions~\cite{varma2023average}. 
Specifically, we reduce the maximum cut problem, where stability is defined with respect to edge deletions, to the max ones problem, where stability is measured with respect to weight changes.

\subsection{Related Work}

% \subsubsection{Lipschitz Continuity}
Lipschitz continuity of discrete algorithms is defined by Kumabe and Yoshida~\cite{KumabeY23}, and they provided Lipschitz continuous approximation algorithms for the minimum spanning tree, shortest path, and maximum weight matching problems.
In a subsequent work~\cite{KumabeY25}, they further provided such algorithms for the minimum weight vertex cover, minimum weight set cover, and minimum weight feedback vertex set problems.
In another work~\cite{KumabeY24}, they defined Lipschitz continuity for allocations in cooperative games and provided Lipschitz continuous allocation schemes for the matching game and the minimum spanning tree game.

A variant known as \emph{pointwise Lipschitz continuity} has also been studied, which is defined using the unweighted Hamming distance instead of the weighted Hamming distance.
Kumabe and Yoshida~\cite{KumabeY23} defined pointwise Lipschitz continuity and provided pointwise Lipschitz continuous algorithms for the minimum spanning tree and maximum weight bipartite matching problems.
Liu et al.~\cite{liu2024pointwise} proposed the \emph{proximal gradient method} as a general technique for solving LP relaxations stably. Using this, they provided pointwise Lipschitz continuous approximation algorithms for the minimum vertex $(S,T)$-cut, densest subgraph, maximum weight ($b$-)matching, and packing integer programming problems. 
% \snote{Is this accurate?}

(Average) sensitivity, introduced by Varma and Yoshida~\cite{varma2023average} with preliminary ideas by Murai and Yoshida~\cite{murai2019sensitivity}, is a notion similar to Lipschitz continuity. While Lipschitz continuity evaluates an algorithm's stability against unit changes in the input weights, (average) sensitivity evaluates an algorithm's stability against (random) deletions of elements from the input. 
Algorithms with small (average) sensitivity have been studied for several problems, such as maximum matching~\cite{varma2023average,YoshidaZ21}, minimum cut~\cite{varma2023average}, knapsack problem~\cite{KumabeY22}, Euclidean $k$-means~\cite{yoshida2022average}, spectral clustering~\cite{peng2020average}, and dynamic programming problems~\cite{kumabe2022average}.
Recently, Fleming and Yoshida~\cite{fleming2024sensitivity} constructed a PCP framework to prove the sensitivity lower bound for the constraint satisfaction problem.

% \ynote{mention Average and Worst-Case Sensitivity}\snote{wrote something}
% \subsubsection{Average and Worst-Case Sensitivity}

%\subsubsection{Algorithmic Meta-Theorems}
%Courcelle's theorem~\cite{Courcelle90, Courcelle92},
%linear optimization version~\cite{ArnborgLS91},
%semiring version~\cite{CourcelleM93},

%\tnote{言及するべきか迷う: clique-width~\cite{CourcelleMR00}}

\subsection{Organization}

The rest of this paper is organized as follows.  
In \Cref{sec:prelim}, we describe the necessary preliminaries on tree decomposition (\Cref{sec:prelim_decomposition,sec:HR-algebra}), logic (\Cref{sec:prelim_logic}), and Lipschitz continuity (\Cref{sec:prelim:lipschitz}).
In \Cref{sec:DP}, we prove \Cref{thm:alg_max} by providing a Lipschitz continuous algorithm for the \MSOt maximization problem. Since the algorithm and analysis for \MSOt minimization are similar, we defer the proof of \Cref{thm:alg_min} to \Cref{app:minimization}.  
In \Cref{sec:special_cases}, we give a more precise Lipschitzness analysis for specific formulas $\varphi$, which includes the proof of  \Cref{thm:alg_independentset}.
In \Cref{sec:planar}, we prove \Cref{thm:baker} by providing a Lipschitz continuous version of Baker's technique.  
In \Cref{sec:lower_bound}, we discuss a lower bound on the Lipschitz constant for general graphs, which justifies focusing on graphs with bounded treewidth.  
In \Cref{sec:DP_clique}, we prove \Cref{thm:alg_max_clique} by providing a Lipschitz continuous algorithm for the \MSOo maximization problem, while we defer the proof of \Cref{thm:alg_min_clique} to \Cref{app:minimization}.

\section{Preliminaries}\label{sec:prelim}

\subsection{Tree Decomposition}\label{sec:prelim_decomposition}

Let $G$ be a graph with $n$ vertices. 
A pair $(\mathcal{B},\mathcal T)$ consisting of a family $\mathcal{B}$ of subsets (called \emph{bags}) of $V(G)$ and a rooted tree $\mathcal T$ whose vertex set is $\mathcal{B}$ is a \emph{(rooted) tree decomposition} of $G$ if it satisfies the following three conditions.
\begin{itemize}
    \itemsep=0pt
    \item $\bigcup_{B\in \mathcal{B}}B = V(G)$.
    \item For each edge $e\in E(G)$, there is a bag $B\in \mathcal{B}$ such that $e\subseteq B$.
    \item For each vertex $v\in V(G)$, the set of bags $\{B\in \mathcal{B}\colon v\in B\}$ induces connected subgraph in $T$.
\end{itemize}
For a tree decomposition $(\mathcal B, \mathcal T)$, we may refer to the \emph{root node}, \emph{leaf nodes}, and the \emph{height} of $\mathcal T$ as those of $(\mathcal B, \mathcal T)$, respectively.
Moreover, $(\mathcal B, \mathcal T)$ is binary if $\mathcal T$ is a binary tree.
The \emph{width} of a tree decomposition $(\mathcal{B}, \mathcal{T})$ is the maximum size of a bag in $\mathcal{B}$ minus one.
The \emph{treewidth} of $G$ is the minimum possible width among all possible tree decompositions of $G$.
It is known that the tree decomposition of $G$ of width at most $2k+1$ can be computed in $2^{O(k)}n$ time~\cite{Korhonen21}, where $k$ is the treewidth of $G$.
Moreover, any tree decomposition of $G$ of width $k$ can be transformed into a binary tree decomposition of width at most $3k+2$ and height at most $O(\log n)$~\cite{BodlaenderH98}.
Thus, we obtain the following lemma.
\begin{lemma}[\cite{BodlaenderH98, Korhonen21}]\label{lem:td_boundeddiam}
Let $G$ be a graph with $n$ vertices and $k$ be the treewidth of $G$. 
Then, a binary tree decomposition of $G$ with width $O(k)$ and height $O(\log n)$ can be computed in $2^{O(k)} n$ time.
\end{lemma}

\subsection{HR-algebra}\label{sec:HR-algebra}

We introduce the notions of \HR-algebra\footnote{The acronym HR stands for hyperedge replacement~\cite{CourcelleE2012}.}, which is one of the algebraic definitions of tree decomposition. 
A \emph{\kgraph{$k$}} is a tuple $G =  \angles{V, E, \src}$ of a set of vertices $V$, a set of edges $E$, and a $k$-vector $\src \in \parens{V \cup \{\nil\}}^k$. We write the $i$-th element of a vector $\src$ by $\src(i)$.
If $\src(i) = u$, we say that $u$ is a \emph{$i$-source}.
We write the set $\{\src(i) : i \in [k]\}\setminus \{\nil\}$ by $\src(G)$.

Let $G$ and $H$ be \kgraphs{$k$}. 
The \emph{parallel-composition} of $G$ and $H$, denoted by $G \paracom H$, is the graph generated by the following procedure.
First, create the disjoint union of $G$ and $H$, and identify the (non-$\nil$) vertices $\src_G(i)$ and $\src_H(i)$ for each $i \in [k]$.
Finally, remove all the self-loops and multi-edges.
Let $B$ be a non-empty subset of $[k]$.
The \emph{source forgetting operation} $\fg_B$ is the function that maps a \kgraph{$k$} $G$ to a \kgraph{$k$} $G'$ such that $V_G = V_{G'}$, $E_G = E_{G'}$, and $\src_{G'}(i) = \nil$ if $i \in B$, and $\src_{G'}(i) = \src_G(i)$ otherwise.

% Let $\pi$ is a permutation of $[k]$.
% The \emph{source renaming operation} $\ren_\pi$ is the function that maps a \kgraph{$k$} $G$ to a \kgraph{$k$} $G'$ such that 
% 	$V_G = V_{G'}$, $E_G = E_{G'}$, and $\src_{G'}(i) = \src_G(\pi(i))$.

\begin{definition}[\HR-algebra]
	A \emph{term} of a \HR-algebra over \kgraphs{{$k$}} is either 
	\begin{itemize}
        \itemsep=0pt
		\item a constant symbol $\Cons{i}$, $\Cons{ij}$, or $\emptyset$ denoting
		 %a graph consisting of only an $i$-source (i.e.,
          a $k$-graph $\langle \{v\},\emptyset,\src\rangle$ with $\src(i)=v$,
		 %a graph consisting of only an edge between an $i$-source and a $j$-source
          a $k$-graph $\langle \{u,v\},\{\{u,v\}\},\src\rangle$ with $\src(i)=u$ and $\src(j)=v$, or
		 the empty graph, respectively;
		\item $t \paracom s$ for any terms $t$ and $s$;
		\item $\fg_B(t) $ for any term $t$ and any $B \subseteq [k]$.
		% \item $\ren_\pi(t) $ for any term $t$ and any permutation $\pi$ of $[k]$.
	\end{itemize}
\end{definition}
We may associate a term with the graphs represented by it.
A term of an \HR-algebra $t$ can be decomposed into the parse tree in the usual way,
called the \emph{(rooted) \HR-parse tree} of $t$.
The \emph{height} of a \HR-parse tree is the maximum distance from the root to a leaf.
It is known that tree decompositions are equivalent to \HR-parse trees in the following sense. 
\begin{proposition}[see e.g. \cite{CourcelleE2012}]
	The treewidth of a graph is at most $k$ if and only if
		the graph can be denoted by a term of a \HR-algebra over \kgraphs{$(k+1)$}.
		\label{prop:equivalence-hr-tw}
\end{proposition}
Moreover, binary tree decomposition can be transformed into an \HR-parse tree with approximately the same height.
\begin{proposition}[$\apxmark$]
	Given a rooted binary tree decomposition $(\mathcal B, \mathcal T)$ of $G$ with width $k$ and height $h$,
	we can compute a \HR-parse tree over \kgraph{$(k+1)$} denoting $G$ with height $O(\log k + h)$.
	\label{prop:height-equivalence}
\end{proposition}
%The proof of \cref{prop:height-equivalence} is essentially the same as the proof of \cref{prop:equivalence-hr-tw} but we give in the full version.
The proof of \cref{prop:height-equivalence} is essentially the same as the proof of \cref{prop:equivalence-hr-tw} but we give in \Cref{app:logics}.

\subsection{Monadic Second-Order Logic}\label{sec:prelim_logic}

% \tatsuya{chotto asshuku sitemo iikamo (my TODO)}
We provide a slightly less formal definition of \emph{monadic second-order (\MSO) logic} over \kgraphs{$k$} for accessibility (see, e.g., \cite{Libkin04} for more formal definition).
% A \emph{first-order term} of \MSO logic is either a vertex variable symbol $v$ or a source symbol $\src(i)$, for $i\in [k]$.
% A \emph{(monadic) second-order term} of \MSO logic is a vertex-set variable symbol $X$.
% The atomic formulas are $s=t$, $\adj(s,t)$ for first-order terms $s$ and $t$, $t \in X$ for first-order term $t$, second-order term $X$, and Boolean constants $\true$ and $\false$. %\ynote{deleted some and's for readability. Revert it if wrong.}
% Here, $\adj(s,t)$ is the predicate that represents whether $s$ and $t$ is adjacent.
%
An \emph{atomic formula} is a formula of the form $s=t$, $\adj(s,t)$, $t \in X$, and Boolean constants $\true$ and $\false$. 
Here, $\adj(s,t)$ is the predicate that represents whether $s$ and $t$ are adjacent.
A \emph{monadic second-order formula} over \kgraphs{$k$} is a formula built from atomic formulas, 
 the usual Boolean connectives $\land, \lor, \lnot, \to$,
 first-order quantifications $\exists x$ and $\forall x$ (quantifications over vertices), 
 and second-order quantifications $\exists X$ and $\forall X$ (quantification over sets of vertices). 
The notation $\varphi(X)$ indicates that $\varphi$ has exactly one free variable $X$. 
The \emph{quantifier height} (or \emph{quantifier rank}) of a formula $\varphi$, denoted by $|\varphi|$, is the maximum depth of quantifier nesting in $\varphi$ (see e.g.~\cite{Libkin04} for a formal definition). 
The \emph{counting monadic second-order (\CMSO{}{}) logic} is an expansion of \MSO logic that have 
  additional predicates $\Card_{m,r}(X)$ for a second-order variable $X$, meaning $|X| \equiv m \pmod r$.
A \CMSO{r}{q}-formula is a \CMSO{}{} formula $\varphi$ such that the
 quantifier height of $\varphi$ is at most $q$ and $r' \leq r$ holds for any predicate $\Card_{m,r'}$ in $\varphi$.
For a logical formula $\varphi$ and a graph $G$, we write $G \models \varphi$ when graph $G$ satisfies the property expressed by $\varphi$.

There are two variants of \MSO logic: \MSOo (or simply \MSO), which is described above, allows quantification over vertices and sets of vertices only, and \MSOt %(also called guarded monadic second-order logic, GSO),
which allows quantification over edges and sets of edges as well.
It is well known that, for any graph $G$ and \MSOt-formula $\varphi$, there exists an \MSOo formula $\varphi'$ such that $G\models \varphi \iff G' \models \varphi'$, where $G'$ is the \emph{incidence graph}, obtained by subdividing each edge of $G$ (with two colors meaning the original and the subdivision vertices)~(see e.g.~\cite{Kreutzer11}). 
It is easy to see that the treewidth of the incidence graph is at most the treewidth of the original graph.
Thus, since this paper focuses on graphs with bounded treewidth, we mainly consider \MSOo.
However, all the theorems are held for \MSOt and graphs of bounded treewidth.
%%%%%%
\newcommand{\frest}[1]{\mathord{{\restriction}_{#1}}}

We introduce some more notations.
For a set $V$, two families $\mathcal A \subseteq 2^V$ and $\mathcal B \subseteq 2^V$ are \emph{separated} 
  if $A \cap B = \varnothing$ for all $A \in \mathcal A$, and $B\in\mathcal B$.
We write $A \cup B$ by $A \dunion B$ if $A \cap B = \varnothing$,
  and  $\mathcal A \boxtimes \mathcal B \coloneqq \{A \cup B : A \in \mathcal A, B \in \mathcal B\}$ if $\mathcal A$ and $\mathcal B$ are separated.
For a formula $\varphi(X)$ and a graph $G$, the set of vertex sets satisfying $\varphi$ is denoted by $\sat(G, \varphi)$,
 that is, $\sat(G,\varphi) = \{ A : G \models \varphi(A)\}$.

The following theorem is a key of the \MSO model-checking algorithm~\cite{Courcelle90, CourcelleE2012}.
\begin{theorem}[\cite{CourcelleE2012}]\label{thm:mso-fg-split}
    Let $r, q, k$ be positive integers.
    Let $\varphi({X})$ be a \CMSO{r}{q}-formula over \kgraphs{$k$} with a second-order free variable ${X}$. Then, the following hold.
    \begin{enumerate}
    \item  For any $B \subseteq [k]$, there exists a \CMSO{r}{q}-formula $\psi({X})$ 
    such that for any \kgraph{$k$} $G$, we have 
    \[\sat(\fg_B(G), \varphi)  =  \sat(G, \psi).\]
    \item There exists a family of tuples
      $\braces{\theta_i({X}), \psi_i({X})}_{i \in [p]}$
      of \CMSO{r}{q}-formulas with a free variable ${X}$
      such that, for any \kgraphs{$k$} $G$ and $H$,
    \[
        \sat(G\paracom H, \varphi) =  \biguplus_{i \in [p]} 
         \{S \cup P : S \in \sat(G, \theta_i),  P \in \sat(H, \psi_i)\}.
    \]
    \end{enumerate}
\end{theorem}
It is known that, for any $r,q,k$, up to logical equivalence, 
there are only finitely many different \CMSO{r}{q}-formulas $\varphi(X)$ over \kgraphs{$k$}~(see e.g., \cite{Libkin04}).
Thus, we can obtain a linear-time algorithm, based on dynamic programming over trees, for \CMSO{r}{q} model checking over bounded treewidth.
Courcelle and Mosbah~\cite{CourcelleM93} adjusted \cref{thm:mso-fg-split} to address optimization, counting, and other problems.
We use a slightly modified version of the theorem of Courcelle and Mosbah.

Let $\nosrc(X) \equiv \bigwedge_{i\in[k]}(\src(i)\not \in X)$, and 
$\varphi\frest{S}(X)$ denote $\varphi(X\cup S) \land \nosrc(X)$ for a set $S\subseteq \src(G)$. %\ynote{here $src_G$ means the set $\{src_G(i) : i \in [k]\}$?}
%\tnote{I define the notation $\src(G)$ in the definition of the HR algebra earlier and use it here.}
\begin{observation}\label{obs:src-splitting}
    Let $G$ be a \kgraph{$k$}, and $\varphi(X)$ be a formula with a free set variable $X$. 
    Then, we have
    $\sat(G,\varphi) = \biguplus_{S\subseteq \src(G)}\sat(G,\varphi \frest{S})\boxtimes \{S\}$.
\end{observation}

%The proof of \cref{cor:separeted-forget-paracom} is given in the full version.
The proof of \cref{cor:separeted-forget-paracom} is given in \Cref{app:logics}

\begin{corollary}[$\apxmark$]\label{cor:separeted-forget-paracom}
    Let $r,q,k$ be positive integers.
    Let $\varphi({X})$ be a \CMSO{r}{q}-formula over \kgraphs{$k$} with a second-order free variables ${X}$. 
    Then, for all $S \subseteq \src(G) $, the following hold.
\begin{enumerate}
    \itemsep=0pt
  \item  For any $B \subseteq [k]$, there exists a \CMSO{r}{q}-formula $\psi({X})$ such that for any \kgraph{$k$} $G$, we have 
  \[
    \sat(\fg_B(G), \varphi\frest{S}) = \biguplus_{S'\subseteq B}\sat(G,\psi\frest{S\cup S'})\boxtimes \{S'\}.
  \]
  \item There exists a family of tuples
      $\braces{\theta_i\frest{S_i}({X}), \psi_i\frest{S'_i}({X})}_{i \in [p]}$
      of \CMSO{r}{q}-formulas with a free variable ${X}$
      such that, for any \kgraphs{$k$} $G$ and $H$,
  \[
    \sat(G\paracom H, \varphi \frest{S}) = \biguplus_{i\in [p]}\sat(G,\theta_i\frest{S_i})\boxtimes \sat(H,\psi_i\frest{S'_i}).
  \]
  Moreover, $S_i \cup S'_i = S$ for all $i\in[p]$. 
\end{enumerate}
\end{corollary}
Then, we can design an algorithm for an (\ClassNameFont{C})\MSO maximization (minimization) problem over graphs of bounded treewidth.
For simplicity, we assume that the given graph $G$ has no sources, that is, $\src(G) = \varnothing$.
Let $t$ be a term denoting $G_t$ and $\varphi(X)$ be the (\ClassNameFont{C})\MSO-formula describing the constraint of the problem.
We recursively compute the value $\OPT(t',\varphi'\frest{S}) = \max\{w(A) : G \models \varphi'\frest{S}(A)\}$ for any subterm $t'$ of $t$ and \MSO-formula $\varphi'\frest{S}$ as follows,
 where $w\in \mathbb{R}_{\geq 0}^{V}$ is the given weight vector and $w(X) = \sum_{x \in X}w_x$ for any $X \subseteq V(G)$.
If $t'$ is of the form $\fg_B(t'')$, then $\OPT(\fg_B(t''), \varphi'\frest{S}) = \max_{S' \subseteq B}\{\OPT(t'', \psi\frest{S\cup S'}) + w(S')\}$,
  where, $\psi$ is the formula obtained from \cref{cor:separeted-forget-paracom}.
If $t'$ is of the form $t_1 \paracom t_2$, then $\OPT(t', \varphi'\frest{S}) = \max_{i\in[p]}\{\OPT(t_1, \theta\frest{S_i}) + \OPT(t_2, \psi\frest{S'_i})\}$,
  where, $\braces{\theta_i\frest{S_i}({X}), \psi_i\frest{S'_i}({X})}_{i \in [p]}$ are the formulas obtained from \cref{cor:separeted-forget-paracom}.
Then, $\OPT(t, \varphi\frest{\varnothing})$ is the maximum value for the problem since $t$ has no sources.
In the following, for simplicity, we may omit the $\frest{S}$ notation and consider only formulas $\varphi$ such that $\sat(G, \varphi)$ does not contain any sources.

In the above setting, the graphs are considered without any special vertex sets (called as colors) or special vertices (called as labels).
However, for $k$-graphs with a constant number of colors and/or labels, an appropriate \HR-algebra and logics can be defined, and the similar results for \cref{cor:separeted-forget-paracom} hold~(see e.g. \cite{CourcelleE2012,Libkin04}).

\subsection{Lipschitz Continuity}\label{sec:prelim:lipschitz}
%\snote{copy-and-pasted from~\cite{KumabeY23}}

We formally define Lipschitz continuity of algorithms.
As this is sufficient for this work, we only consider algorithms for vertex-weighted graph problems.
The definition can be naturally extended to other settings, such as edge-weighted problems.
Let $G=(V,E)$ be a graph.
For vertex sets $X,X'\subseteq V$ and weight vectors $w,w' \in \mathbb R_{\geq 0}^V$, we define the \emph{weighted Hamming distance} between $(X,w)$ and $(X',w')$ by
\begin{align*}
    d_{\mathrm{w}}((X,w),(X',w')) 
    := \left\|\sum_{v\in X}\bm{1}_vw_v - \sum_{v\in X'}\bm{1}_vw'_v\right\|_1 
    = \sum_{v\in X\cap X'}|w_v - w'_v|+ \sum_{v\in X\setminus X'}w_v + \sum_{v\in X'\setminus X}w'_v,
\end{align*}
where $\bm{1}_v\in \{0,1\}^V$ denotes the \emph{characteristic vector} of $v$, that is, the vector $\bm{1}_{v,u}=1$ holds if and only if $u=v$.
For two probability distributions $\mathcal{X}$ and $\mathcal{X}'$ over subsets of $V$, we define
\begin{align*}
    \EM\left((\mathcal{X},w),(\mathcal{X}',w')\right)
    := \inf_{\mathcal{D}}\E_{(X,X')\sim \mathcal{D}}\left[d_{\mathrm{w}}((X,w),(X',w'))\right],
\end{align*}
where the minimum is taken over \emph{couplings} of $\mathcal{X}$ and $\mathcal{X}'$, that is, distributions over pairs of sets such that its marginal distributions on the first and second coordinates are equal to $\mathcal{X}$ and $\mathcal{X}'$, respectively.
Consider an algorithm $\mathcal{A}$, that takes a graph $G=(V,E)$ and a weight vector $w\in \mathbb{R}_{\geq 0}^{V}$ as an input and outputs a vertex subset $X\subseteq V$.
We denote the output distribution of $\mathcal{A}$ for weight $w$ as $\mathcal{A}(w)$.
The \emph{Lipschitz constant} of $\mathcal{A}$ is defined by
\begin{align*}
    \sup_{\substack{w,w'\in \mathbb{R}_{\geq 0},\\w\neq w'}}\frac{\EM\left((\mathcal{A}(w),w),(\mathcal{A}(w'),w')\right)}{\|w-w'\|_1}.
\end{align*}
% We say that an algorithm is \emph{scale-invariant} if it outputs the same for weight vectors $w \in \mathbb R^V_{\geq 0}$ and $\alpha w$ for any $\alpha \geq 0$.
% Then, the Lipschitz constant of a scale-invariant algorithm is bounded.

For two random variables $\bm{X}$ and $\bm{X}'$, the \emph{total variation distance} between them is given as:
\begin{align*}
    \TV\left(\bm{X},\bm{X}'\right):=\inf_{\mathcal{D}}\Pr_{(X,X')\sim \mathcal{D}} \left[X\neq X'\right],
\end{align*}
where the minimum is taken over couplings between $\bm{X}$ and $\bm{X}'$, that is, distributions over pairs such that its marginal distributions on the first and the second coordinates are equal to $\bm{X}$ and $\bm{X}'$, respectively.
For an element $u \in V$, we use $\bm{1}_u\in \mathbb{R}_{\geq 0}^V$ to denote the \emph{characteristic vector} of $u$, that is, $\bm{1}_u(u) = 1$ and $\bm{1}_u(v) = 0$ for $v \in V \setminus \{u\}$.
The following lemma indicates that, to bound the Lipschitz constant, it suffices to consider pairs of weight vectors that differ by one coordinate.
\begin{lemma}[\rm{\cite{KumabeY23}}]\label{lem:seeoneelement}
Suppose that there exist some $c>0$ and $L>0$ such that
\[
    \EM\left((\mathcal{A}(G,w),w), (\mathcal{A}(G,w+\delta \mathbf{1}_u),w+\delta \mathbf{1}_u)\right)\leq \delta L
\]
holds for all $w\in \mathbb{R}_{\geq 0}^V$, $u\in V$ and $0<\delta\leq c$. Then, $\mathcal{A}$ is $L$-Lipschitz.
\end{lemma}

In our algorithm, we use the following procedures \emph{softmax} and \emph{softmin}.
These procedures are derived from the \emph{exponential mechanism} in the literature of \emph{differential privacy}~\cite{Dwork06} and frequently appear in the literature on Lipschitz continuity~\cite{KumabeY23,KumabeY25}.
Here, we organize them into a more convenient form for our use.
Let $p\in \mathbb{Z}_{\geq 1}$, $x_1,\dots, x_p\in \mathbb{R}_{\geq 0}$, and $\epsilon \in (0,1]$.
The \emph{softmax} of $x_1,\dots, x_p$ is taken as follows.
If $\max_{i\in [p]}x_i=0$, we set $\softargmax^{\epsilon}_{i\in [p]}x_i$ to be an arbitrary index $i$ with $x_i=0$ and $\softmax^{\epsilon}_{i\in [p]}x_i=0$.
Assume otherwise.
First, we sample $\bm{c}$ uniformly from $\left[\frac{2\epsilon^{-1}\log(2p\epsilon^{-1})}{\max_{i\in [p]}x_i},\frac{4\epsilon^{-1}\log(2p\epsilon^{-1})}{\max_{i\in [p]}x_i}\right]$.
Let $\bm{i}^*$ be a probability distribution over $[p]$ such that
\begin{align*}
    \Pr\left[\bm{i}^*=i\right]=\frac{\exp(\bm{c}x_i)}{\sum_{i'\in [p]}\exp(\bm{c}x_{i'})}.
\end{align*}
holds for all $i\in [p]$.
Then, we define $\softargmax^{\epsilon}_{i\in [p]}x_i:=\bm{i}^*$ and $\softmax^{\epsilon}_{i\in [p]}x_i:=x_{\bm{i}^*}$.
We have the following.
The proofs of \Cref{lem:softmax_approx,lem:softmax_tv} are given in \Cref{app:softmax}.
%The proofs of \Cref{lem:softmax_approx,lem:softmax_tv} are given in the full version.
\begin{lemma}\label{lem:softmax_approx}
We have $\E\left[\softmax^{\epsilon}_{i\in [p]}x_i\right]\geq (1-\epsilon)\max_{i\in [p]}x_i$.
\end{lemma}

\begin{lemma}\label{lem:softmax_tv}
Let $\delta > 0$ and assume $\delta \leq \frac{\max_{i\in [p]}x_i}{4\epsilon^{-1}\log(2p\epsilon^{-1})}$.
Let $x_1',\dots, x_p'\in \mathbb{R}_{\geq 0}$ be numbers such that $x_i\leq x_i' \leq x_i + \delta$ holds for all $i\in [p]$.
Then, $\TV\left(\softargmax^{\epsilon}_{i\in [p]}x_i, \softargmax^{\epsilon}_{i\in [p]}x_i'\right)\leq \frac{10\epsilon^{-1}\log(2p\epsilon^{-1})\delta}{\max_{i\in [p]}x_i}$.
\end{lemma}

Similarly, the \emph{softmin} of $x_1,\dots, x_p$ is taken as follows.
If $\min_{i\in [p]}x_i=0$, we set $\softargmin^{\epsilon}_{i\in [p]}x_i$ be an arbitrary index $i$ with $x_i=0$ and $\softmin^{\epsilon}_{i\in [p]}x_i=0$.
Assume otherwise.
First, we sample $\bm{c}$ uniformly from $\left[\frac{2\epsilon^{-1}\log(2p\epsilon^{-1})}{\min_{i\in [p]}x_i},\frac{4\epsilon^{-1}\log(2p\epsilon^{-1})}{\min_{i\in [p]}x_i}\right]$.
Let $\bm{i}^*$ be a probability distribution over $[p]$ such that
\begin{align*}
    \Pr\left[\bm{i}^*=i\right]=\frac{\exp(-\bm{c}x_i)}{\sum_{i'\in [p]}\exp(-\bm{c}x_{i'})}.
\end{align*}
holds for all $i\in [p]$.
Then, we define $\softargmin^{\epsilon}_{i\in [p]}x_i:=\bm{i}^*$ and $\softmin^{\epsilon}_{i\in [p]}x_i:=x_{\bm{i}^*}$.
We have the following.
The proofs of \Cref{lem:softmin_approx,lem:softmin_tv} are given in \Cref{app:softmax}.
%The proofs of \Cref{lem:softmin_approx,lem:softmin_tv} are given in the full version.
\begin{lemma}\label{lem:softmin_approx}
We have $\E\left[\softmin^{\epsilon}_{i\in [p]}x_i\right]\leq (1+\epsilon)\min_{i\in [p]}x_i$.
\end{lemma}

\begin{lemma}\label{lem:softmin_tv}
Let $\delta > 0$ and assume $\delta \leq \frac{\min_{i\in [p]}x_i}{4\epsilon^{-1}\log(2p\epsilon^{-1})}$.
Let $x_1',\dots, x_p'\in \mathbb{R}_{\geq 0}$ be numbers such that $x_i\leq x_i' \leq x_i + \delta$ holds for all $i\in [p]$.
Then, $\TV\left(\softargmin^{\epsilon}_{i\in [p]}x_i, \softargmin^{\epsilon}_{i\in [p]}x_i'\right)\leq \frac{10\epsilon^{-1}\log(2p\epsilon^{-1})\delta}{\min_{i\in [p]}x_i}$.
\end{lemma}
\section{Lipschitz Continuous Algorithms for \MSOt Optimization Problems}\label{sec:DP}

In this section, we prove \Cref{thm:alg_max,thm:alg_min,thm:alg_independentset}. Since the proofs are similar to that for \Cref{thm:alg_max}, we provide most of the proof of \Cref{thm:alg_min} in \Cref{app:minimization}.
%In this section, we prove \Cref{thm:alg_max,thm:alg_min,thm:alg_independentset}. Since the proofs are similar to that for \Cref{thm:alg_max}, we provide most of the proof of \Cref{thm:alg_min} in the full version.
In \Cref{sec:dp_definition}, we define the notations and give a brief overview of the algorithm.
\Cref{sec:dp_intro} handles the base case, where the graph consists of a single vertex.
\Cref{sec:dp_paracom,sec:dp_forget} analyze the impact of the parallel composition and the forget operations, respectively, on approximability and the Lipschitz constant.
In \Cref{sec:dp_combine}, we put everything together to prove \Cref{thm:alg_max}.
Finally, in \Cref{sec:dp_independent}, we provide a more refined analysis for the special case of the maximum weight independent set problem to prove \Cref{thm:alg_independentset}.

% \ynote{We need some explanation about the organization of this section. For example, the section ``Introduce'' sounds too abrupt.}\snote{wrote}
\subsection{Definitions}\label{sec:dp_definition}

For a graph $G$, a weight vector $w$, and an \MSOt formula $\varphi$, we define  
\begin{align}
    \OPT_w[G,\varphi] = \max_{S\subseteq V(G), G\models \varphi(S)}w(S).\label{eq:def_opt_max}
\end{align}
for maximization problems and
\begin{align}
    \OPT_w[G,\varphi] = \min_{S\subseteq V(G), G\models \varphi(S)}w(S).\label{eq:def_opt_min}
\end{align}
for minimization problems.
If $\sat(G,\varphi)=\emptyset$, we define $\OPT_w[G,\varphi]=\nil$.
For graphs $G$ and \MSOt formulas $\varphi$ with $\sat(G,\varphi)\neq \emptyset$, our algorithm recursively computes the vertex set $\DP_w[G,\varphi]$ that approximately achieves the maximum in Equation~(\ref{eq:def_opt_max}) or minimum in Equation~(\ref{eq:def_opt_min}). 
%If no such set exists, we define $\DP_w[G,\varphi]=\nil$.
When it is clear from the context, we omit the subscript $w$ and simply write $\OPT[G,\varphi]$ and $\DP[G,\varphi]$.
We perform the dynamic programming algorithm using the formulas defined in~\Cref{cor:separeted-forget-paracom} over the parse tree of a term obtained from \Cref{lem:td_boundeddiam} and \Cref{prop:height-equivalence}. %\snote{言葉遣いが怪しいか？}
Specifically, at every stage of the algorithm, it is guaranteed that each $X\in \sat(G,\varphi)$ satisfies $X\subseteq V(G)\setminus \src(G)$.
Furthermore, our algorithm is randomized. Thus, $\DP[G,\varphi]$ can be considered as a probability distribution over vertex sets in $\sat(G,\varphi)$.

To bound the approximation ratio, for a weight vector $w\in \mathbb{R}_{\geq 0}^V$, we bound the ratio between $\E\left[w(\DP[G,\varphi])\right]$ and $\OPT[G,\varphi]$.
To bound the Lipschitz constant, for a weight vector $w\in \mathbb{R}_{\geq 0}^V$, $u\in V$, and $\delta > 0$, we bound $\EM\left(\DP_w[G,\varphi],\DP_{w+\delta \bm{1}_u}[G,\varphi]\right)$.
%\begin{align*}
%    \EM\left(\DP_w[G,\varphi],\DP_{w+\delta \bm{1}_u}[G,\varphi]\right).
%\end{align*}
From now on, we will concentrate on maximizing problems.
The algorithm for minimization problems is similar to that for maximization problems, while the detail of the analysis is slightly differ.
We discuss the minimization version in Appendix~\ref{app:minimization}.
%We discuss the minimization version in the full version.

%\snote{$\nil$ に対して何かを定義しようとすると自動的に無視されますみたいなことを書く。$\OPT[G,\varphi]=\nil$ のところに対してはそもそも呼ばない、みたいな感じにすると楽か？}

\subsection{Base Case}\label{sec:dp_intro}
Here we consider the base case.
Let $G$ be a graph with a single vertex $v$ and no edges, where $v$ is a source of $G$, and $\varphi$ be an \MSOt formula such that $\sat(G,\varphi)$ is nonempty and contains no set containing $v$.
In particular, we have $\sat(G, \varphi)=\{\emptyset\}$.
We set $\DP[G,\varphi] := \emptyset$.
Since $w(\emptyset)=0$, it is clear that $w\left(\DP[G,\varphi]\right) = \OPT[G,\varphi] = 0$.
Furthermore, it is obvious that
\begin{align*}
    \EM\left(\DP_w[G,\varphi],\DP_{w+\delta \bm{1}_u}[G,\varphi]\right)=\EM(\emptyset,\emptyset) = 0.
\end{align*}
%\ynote{what do we do when $v$ appears in  $\sat(G,\varphi)$?}
%\snote{we (want to) ensure $v$ does not appear in $\sat(G,\varphi)$. This is why the modification of DP is needed.}

\subsection{Parallel Composition}\label{sec:dp_paracom}

Next, we consider the parallel composition.
Let
\begin{align}
    \sat(G\paracom H, \varphi) =  \biguplus_{i\in [p]}\sat(G, \theta_{i})  \boxtimes \sat(H, \psi_{i}).\label{eq:dp_paracom}
\end{align}
Assume we have already computed $\DP[G, \theta_{i}]$ and $\DP[H, \psi_{i}]$ for each $i\in [p]$.
We compute $\DP[G\paracom H, \varphi]$.
For each $i\in [p]$, we denote $\OPT_{i}:=\OPT[G, \theta_{i}] + \OPT[H, \psi_{i}]$.
%\begin{align*}
%    \OPT_{i}:=\OPT[G, \theta_{i}] + \OPT[H, \psi_{i}].
%\end{align*}
%We first sample $\bm{c}\in \mathbb{R}_{>0}$ uniformly from $\left[\frac{2\log (2p\varepsilon^{-1})}{\varepsilon\OPT[G\paracom H,\varphi,S]},\frac{4\log (2p\varepsilon^{-1})}{\varepsilon\OPT[G\paracom H,\varphi,S]}\right]$.
%Let $i\in [p]$. and $j\in [q_i]$, we denote 
%\begin{align*}
%\OPT_i=\OPT[G, \theta_{i},S\cap V(G)] + \OPT[H, \psi_{i},S\cap V(H)].
%\end{align*}
By definition, we have $\max_{i\in [p]}\OPT_i = \OPT[G\paracom H, \varphi]$.
Then, we take $\bm{i}^*=\softargmax^{\varepsilon}_{i\in [p]}\OPT_i$ and define
$\DP[G\paracom H, \varphi] :=  \DP[G, \theta_{\bm{i}^*}] \cup \DP[H, \psi_{\bm{i}^*}]$.
%\begin{align*}
%    \DP[G\paracom H, \varphi] :=  \DP[G, \theta_{\bm{i}^*}] \cup \DP[H, \psi_{\bm{i}^*}].
%\end{align*}
%We define the probability distribution $\bm{i}^*$ over $[p]$ so that 
%\begin{align*}
%\Pr[\bm{i}^*=i]=\frac{\exp\left(\bm{c}\OPT_i\right)}{\sum_{i'=1}^{p}\exp\left(\bm{c}\OPT_i\right)}.
%\end{align*}
%Then we define 
%\begin{align*}
%    \DP[G\paracom H, \varphi,S] :=  \DP[G, \theta_{\bm{i}^*},S\cap V(G)] \cup \DP[H, \psi_{\bm{i}^*},S\cap V(H)].
%\end{align*}
First we analyze the approximation ratio.

\begin{lemma}\label{lem:paracom_approx}
Let $0 < \alpha \leq 1$ and suppose $\E\left[w(\DP[G,\theta_i])\right]\geq \alpha \OPT[G,\theta_i]$ and $\E\left[w(\DP[H,\psi_i])\right]\geq \alpha \OPT[H,\psi_i]$
%\begin{align*}
%    \E\left[w(\DP[G,\theta_i])\right]\geq \alpha \OPT[G,\theta_i] \quad \text{and} \quad 
%    \E\left[w(\DP[H,\psi_i])\right]\geq \alpha \OPT[H,\psi_i]
%\end{align*}
hold for all $i\in [p]$.
Then, we have $\E\left[w\left(\DP[G\paracom H,\varphi]\right)\right]\geq (1-\varepsilon)\alpha\OPT[G\paracom H,\varphi]$.
\end{lemma}
\begin{proof}
We have
\begin{align*}
    & \E\left[w(\DP[G\paracom H,\varphi])\right]
    = \E\left[w(\DP[G, \theta_{\bm{i}^*}]  \cup \DP[H, \psi_{\bm{i}^*}])\right]
    \geq \alpha\E\left[\OPT[G, \theta_{\bm{i}^*}] + \OPT[H, \psi_{\bm{i}^*}]\right] \\
    &= \alpha\E\left[\softmax^{\varepsilon}_{i\in [p]}\OPT_i\right]
    \geq (1-\varepsilon)\alpha\max_{i\in [p]}\OPT_i
    =(1-\varepsilon)\alpha\OPT[G\paracom H,\varphi].\qedhere
\end{align*}
%where the second inequality is from
%\begin{align*}
%    \E[x_{\bm{i}^*}]
%    &\geq \Pr\left[x_{\bm{i}^*}\geq \left(1-\frac{\varepsilon}{2}\right)\max_{i\in [p]} \OPT_i\right]\cdot \left(1-\frac{\varepsilon}{2}\right)\max_{i\in [p]}\OPT_i\\
%    &\geq \left(1-p\cdot \frac{\exp\left(\bm{c}\left(1-\frac{\varepsilon}{2}\right)\max_{i\in[p]}\OPT_i\right)}{\sum_{i\in [p]}\exp(\bm{c}\OPT_i)}\right)\cdot \left(1-\frac{\varepsilon}{2}\right)\max_{i\in [p]}\OPT_i\\
%    &\geq \left(1-p\cdot \frac{\exp\left(\bm{c}\left(1-\frac{\varepsilon}{2}\right)\max_{i\in[p]}\OPT_i\right)}{\max_{i\in [p]}\exp(\bm{c}\OPT_i)}\right)\cdot \left(1-\frac{\varepsilon}{2}\right)\max_{i\in [p]}\OPT_i\\
%    &= \left(1-p\cdot \exp\left(-\frac{\bm{c}\varepsilon}{2}\max_{i\in[p]}\OPT_i\right)\right)\cdot \left(1-\frac{\varepsilon}{2}\right)\max_{i\in [p]}\OPT_i\\
%    &\geq \left(1-p\cdot \exp\left(-\log(2p\varepsilon^{-1})\right)\right)\cdot \left(1-\frac{\varepsilon}{2}\right)\max_{i\in [p]}\OPT_i\\
%    &= \left(1-\frac{\varepsilon}{2}\right)\cdot \left(1-\frac{\varepsilon}{2}\right)\max_{i\in [p]}\OPT_i
%    \geq (1-\varepsilon)\max_{i\in [p]}\OPT_i.\qedhere
%\end{align*}
\end{proof}

Now we analyze the Lipschitz constant.
We denote the variable $\bm{i}^*$ corresponding to the weight $w$ and $w+\delta \bm{1}_u$ by $\bm{i}^*_w$ and $\bm{i}^*_{w+\delta \bm{1}_u}$, respectively.
We note that $\DP_w[G,\theta_i]$ and $\DP_{w+\delta \bm{1}_u}[G,\theta_i]$ are the same as a distribution unless $u\in V(G)\setminus \src_G$. 
The same also holds for $H$.
Since $V(G)\setminus \src_G$ and $V(H)\setminus \src_H$ are disjoint, without loss of generality, we can assume the $\DP_w[H,\psi_i]$ and $\DP_{w+\delta \bm{1}_u}[H,\psi_i]$ are the same as a distribution.
We have the following.
\begin{lemma}\label{lem:paracom_lip}
%Let $c\in \mathbb{R}_{\geq 0}$ and assume $\bm{c}=\bm{c}^{\delta u}=c$. 
Let $\beta\in \mathbb{R}_{\geq 0}$ and suppose $\EM\left(\DP_{w}[G,\theta_i],\DP_{w+\delta\bm{1}_{u}}[G,\theta_i]\right)\leq \beta$
%\begin{align*}
%    \EM\left(\DP_{w}[G,\theta_i],\DP_{w+\delta\bm{1}_{u}}[G,\theta_i]\right)\leq \beta
%\end{align*}
holds for all $i\in [p]$.
Then, we have $\EM\left(\DP_{w}[G\paracom H,\varphi],\DP_{w+\delta\bm{1}_{u}}[G\paracom H,\varphi]\right) \leq 30\varepsilon^{-1}\log(2p\varepsilon^{-1})\delta + \beta$.
%\begin{align*}
%    \EM\left(\DP_{w}[G\paracom H,\varphi],\DP_{w+\delta\bm{1}_{u}}[G\paracom H,\varphi]\right) \leq 30\varepsilon^{-1}\log(2p\varepsilon^{-1})\delta + \beta.
%\end{align*}
\end{lemma}
\begin{proof}
We have
\begin{align*}
    &\EM\left(\DP_{w}[G\paracom H,\varphi],\DP_{w+\delta\bm{1}_{u}}[G\paracom H,\varphi]\right)\\
    &\leq \TV\left(\bm{i}^{*}_w,\bm{i}^{*}_{w+\delta \bm{1}_u}\right)\left(\OPT_w[G\paracom H,\varphi] + \OPT_{w+\delta \bm{1}_u}[G\paracom H, \varphi]\right)
    + \max_{i\in [p]}\left(\EM\left(\DP_{w}[G,\theta_i],\DP_{w+\delta \bm{1}_u}[G,\theta_i]\right)\right)\\
    &\leq \frac{10\varepsilon^{-1}\log (2p\varepsilon^{-1})\delta}{\OPT_w[G\paracom H, \varphi]}\left(\OPT_w[G\paracom H,\varphi] + \OPT_{w+\delta \bm{1}_u}[G\paracom H, \varphi]\right) + \beta
    \leq 30\varepsilon^{-1}\log(2p\varepsilon^{-1})\delta+\beta,
\end{align*}
% \ynote{what's $\OPT[G\paracom H, \varphi,S]$?}\snote{Forgot subscript $w$. Fixed.}
where the last inequality is from $\varepsilon \leq 1$ and $\delta\leq \OPT_w[G\paracom H, \varphi]$.
\end{proof}

\subsection{Forget}\label{sec:dp_forget}

Here we consider forget operation.
Let $B\subseteq \src_G$ and 
  % \[
  %   \sat(\fg_B(G), \varphi\frest{S}) = \biguplus_{S'\subseteq B}\sat(G,\psi\frest{S\cup S'})\boxtimes \{S'\}.
  % \]
\begin{align}
    \sat(\fg_{B}(G), \varphi) =  \biguplus_{i\in [p]}  \sat(G,\varphi_i)\boxtimes \{S_i\}.\label{eq:dp_forget}
\end{align}
Assume we have already computed $\DP[G,\varphi_i]$ for each $i\in [p]$.
We compute $\DP[\fg_B(G),\varphi]$.

We first sample $\bm{c}\in \mathbb{R}_{>0}$ uniformly from $\left[\frac{2\log (2p\varepsilon^{-1})}{\varepsilon\OPT[G,\varphi,S]},\frac{4\log (2p\varepsilon^{-1})}{\varepsilon\OPT[G,\varphi,S]}\right]$.
For each $i\in [p]$, we denote $\OPT_i=\OPT[G, \varphi_{i}] + w(S_i)$.
%\begin{align*}
%    \OPT_i=\OPT[G, \varphi_{i}] + w(S_i).
%\end{align*}
By definition, we have $\max_{i\in [p]}\OPT_i = \OPT[\fg_B(G), \varphi]$.
Then, we take $\bm{i}^*=\softargmax^{\varepsilon}_{i\in [p]}\OPT_i$ and define $\DP[\fg_B(G), \varphi] :=  \DP[G, \varphi_{\bm{i}^*}] \cup S_{\bm{i}^*}$.
%\begin{align*}
%    \DP[\fg_B(G), \varphi] :=  \DP[G, \varphi_{\bm{i}^*}] \cup S_{\bm{i}^*}.
%\end{align*}
First we analyze the approximation ratio.

\begin{lemma}\label{lem:forget_approx}
Let $0 < \alpha \leq 1$ and suppose $\E\left[w(\DP[G,\varphi_i])\right]\geq \alpha \OPT[G,\varphi_i]$
%\begin{align*}
%    \E\left[w(\DP[G,\varphi_i])\right]&\geq \alpha \OPT[G,\varphi_i]
%\end{align*}
holds for all $i\in [p]$.
Then, we have $\DP[\fg_{B}(G),\varphi]\geq (1-\varepsilon)\alpha\OPT[\fg_{B}(S),\varphi]$.
\end{lemma}
\begin{proof}
We have
\begin{align*}
    \E\left[w(\DP[G,\varphi])\right]
    &= \E\left[w(\DP[G, \varphi_{\bm{i}^*}]  \cup S_{\bm{i}^*})\right]
    \geq \E\left[\alpha\OPT[G, \varphi_{\bm{i}^*}] + w(S_{\bm{i}^*})\right]
    \geq \alpha\E\left[\OPT[G, \varphi_{\bm{i}^*}] + w(S_{\bm{i}^*})\right]\\
    & = \alpha\E\left[\softmax^{\varepsilon}_{i\in [p]}\OPT_i\right]
    \geq (1-\varepsilon)\alpha\max_{i\in [p]}\OPT_i
    =(1-\varepsilon)\alpha\OPT[\fg_{B}(G),\varphi].\qedhere
\end{align*}
\end{proof}

Now we analyze the Lipschitz constant.
\begin{lemma}\label{lem:forget_lip}
Let $\beta\in \mathbb{R}_{\geq 0}$ and suppose $\EM\left(\DP_{w}[G,\varphi_{i}],\DP_{w+\delta \bm{1}_u}[G,\varphi_{i}]\right)\leq \beta$
%\begin{align*}
%    \EM\left(\DP_{w}[G,\varphi_{i}],\DP_{w+\delta \bm{1}_u}[G,\varphi_{i}]\right)\leq \beta
%\end{align*}
holds for all $i\in [p]$. Then, we have $\EM\left(\DP_{w}[\fg_{B}(G),\varphi],\DP_{w+\delta\bm{1}_{u}}[\fg_{B}(G),\varphi]\right)\leq 31\varepsilon^{-1}\log(2p\varepsilon^{-1})\delta + \beta$.
%\begin{align*}
%    \EM\left(\DP_{w}[\fg_{B}(G),\varphi],\DP_{w+\delta\bm{1}_{u}}[\fg_{B}(G),\varphi]\right)\leq 31\varepsilon^{-1}\log(2p\varepsilon^{-1})\delta + \beta.
%\end{align*}
% \ynote{$DP[G,\emptyset]$ should be $DP[fg_B(G),\varphi]$ or something?}\snote{yes. fixed}
\end{lemma}
\begin{proof}
We have
\begin{align*}
    \EM&\left(\DP_{w}[\fg_{B}(G),\varphi],\DP_{w+\delta\bm{1}_{u}}[\fg_{B}(G),\varphi]\right)\\
    &\leq \TV\left(\bm{i}^{*}_w,\bm{i}^{*}_{w+\delta \bm{1}_u}\right) \left(\OPT_{w}[\fg_{B}(G),\varphi]+\OPT_{w+\delta \bm{1}_u}[\fg_{B}(G),\varphi]\right)\\
    & \quad + \max_{i\in [p]}\left(\EM\left(\DP_{w}[G,\varphi_i]\cup S_i,\DP_{w+\delta \bm{1}_u}[G,\varphi_i]\cup S_i\right)\right)\\
    &\leq \frac{10\varepsilon^{-1}\log(2p\varepsilon^{-1})\delta}{\OPT_w[\fg_{B}(G),\varphi]}\left(\OPT_w[\fg_{B}(G),\varphi]+\OPT_{w+\delta \bm{1}_u}[\fg_{B}(G),\varphi]\right) \\
    &\quad + \max_{i\in [p]}\left(\EM\left(\DP_{w}[G,\varphi_i],\DP_{w+\delta \bm{1}_u}[G,\varphi_i]\right)\right) + \max_{i\in [p]}\left[d_{\mathrm{w}}((S_i,w),(S_i,w+\delta \bm{1}_u))\right]\\
    &\leq 30\varepsilon^{-1}\log(2p\varepsilon^{-1})\delta + \beta + \delta
    \leq 31\varepsilon^{-1}\log(2p\varepsilon^{-1})\delta + \beta.\qedhere
\end{align*}
\end{proof}

\subsection{Putting Together}\label{sec:dp_combine}

Let $G=(V,E)$ be a graph with treewidth $k$.
From \Cref{lem:td_boundeddiam} and \Cref{prop:height-equivalence}, we can compute a \HR-parse tree $t$ over $(k+1)$-graph denoting $G$ with height $O(\log k + \log n)\leq O(\log n)$.
We perform the dynamic programming algorithm on this parse tree.
We have the following.
% \snote{ちょっと修正しました}
\begin{lemma}\label{lem:combineall}
Let $\varepsilon \in (0,1]$ and $h$ be the height of $t$.
Our algorithm outputs a solution $X\in \sat(G,\varphi)$ such that $\E[w(X)]\geq (1-h\varepsilon)\OPT[G,\varphi]$.
The Lipschitz constant is bounded by $31h\varepsilon^{-1}\log(2p_{\max}\varepsilon^{-1})$, where $p_{\max}$ is the maximum of $p$ among all update formulas~(\ref{eq:dp_paracom})~or~(\ref{eq:dp_forget}) that the algorithm uses.
\end{lemma}
\begin{proof}
The approximability bound is obtained by repeatedly applying \Cref{lem:paracom_approx,lem:forget_approx} and $(1-\varepsilon)^h\geq 1-h\varepsilon$.
The Lipschitzness bound is obtained by repeatedly applying \Cref{lem:paracom_lip,lem:forget_lip}.
\end{proof}

Since $p_{\max}$ is bounded by a function of $k$ and $|\varphi|$, we have the following.
%\begin{theorem}
%Let $G$ be a graph with treewidth $t$, $\varphi$ be an \MSOt formula, and $\varepsilon \in \mathbb{R}_{>0}$.
%Then, there is a $(1-\varepsilon)$-approximation algorithm for \MSOt maximization with Lipschitz constant $(f(t,|\varphi|)+\log \varepsilon^{-1}+\log \log n)\varepsilon^{-1}\log^2 n$ that runs in FPT time parameterized by $t$ and $|\varphi|$.
%\end{theorem}
\begin{proof}[Proof of \Cref{thm:alg_max}]
    The claim follows by substituting $\varepsilon$, $h$, and $k$ in \Cref{lem:combineall} with $\frac{\varepsilon}{\Theta(\log n)}$, $O(\log n)$, and $3k+2$, respectively.
\end{proof}

\section{Special Cases}\label{sec:special_cases}

\subsection{Independent Set}\label{sec:dp_independent}

From the proof of \Cref{thm:alg_max}, $f(k, \varphi)$ in \Cref{thm:alg_max} is bounded by $\log(2p_{\max})$. 
Particularly, if $p_{\max}$ is bounded by $2^{O(k)}$, the Lipschitz constant depends linearly on $k$.
Here, we prove \Cref{thm:alg_independentset} for the maximum weight independent set problem by showing that $p_{\max} \leq 2^{O(k)}$.
Similar arguments can be applied to several other problems, such as the minimum weight vertex cover problem and the minimum weight dominating set problem.
For a $k$-graph $G$ and $S\subseteq \src(G)$, let 
\begin{align*}
    \varphi_{S}(X) \equiv \forall x(x\in \src(G)\rightarrow \lnot x\in X)\land \forall x \forall y (((x\in X\lor x\in S)\land (y\in X\lor y\in S))\rightarrow \lnot \adj(x,y)) .
\end{align*}
In words, $\sat(G,\varphi_S)$ is the family of subsets $X$ of $V$ such that $X$ is disjoint from $\src(G)$ and $X\cup S$ is an independent set of $G$. 
Then, we have
\begin{align*}
    \sat(\fg_B(G), \varphi_S) = \biguplus_{S'\subseteq B}\sat(G,\varphi_{S\cup S'})\boxtimes \{S'\},\quad
    \sat(G\paracom H, \varphi_S) = \sat(G,\varphi_{S})\boxtimes \sat(H,\varphi_S).
\end{align*}
Particularly, we have $p_{\max}\leq 2^{O(k)}$ and therefore \Cref{thm:alg_independentset} is proved.

\subsection{Max Ones}\label{sec:dp_maxone}
Recall that in \emph{max ones problem}, we given a 3CNF formula $\Phi$ over a variable set $X$ and a weight function $w:X \to \mathbb R_{\geq 0}$, the goal is to find a satisfying assignment $\sigma:X \to \{0,1\}$ that maximizes the weight $\sum_{x \in \sigma^{-1}(1)}w(x)$.
We reduce the problem to a graph problem.
Let $G_\Phi = (X \cup \bar{X} \cup \mathcal C, E)$ be the graph such that
$X$ is the variable set of $\Phi$, $\bar{X} = \{\bar{x}_i : x_i\in X\}$, 
$\mathcal C$ is the clause set of $\Phi$, $\{x_i, \bar{x}_i\}\in E$ for all $x_i\in X$, $\{C_j, x_i\}\in E$ iff clause $C_j$ contains $x_i$ as a positive literal, and $\{C_j, \bar{x}_i\}\in E$ iff clause $C_j$ contains $x_i$ as a negative literal.
Let $w'(x) = w(x)$ if $x\in X$, and otherwise $w'(x) = 0$. Note that the treewidth of $G_\Phi$ is at most $2|X|$.

Here, the notations $a \in A \cup B$, $a \in A\cap B$, $a \in A \setminus B$, $\forall a \in A\; \psi$, and $\exists a \in A\; \psi$ are syntactic sugar defined in the usual sense.
For a $k$-graph $G$ and $S, D\subseteq \src(G)$, let 
\[
 \varphi_{S,D}(A) \equiv \exists B  \begin{bmatrix*}[l]
  & \forall a \in A\cup S\; \forall b \in B \; (a \in X \land b \in \bar X \land \lnot\adj(a, b)) \\
  \land & \forall x \in X \cup \bar X\;\forall y \in X \cup \bar X\; (\adj(x,y) \to (x \in A\cup S\cup B \lor y \in A \cup S \cup B)\\
  \land & \forall s \in \src(G) \; (\lnot s \in A) 
  \land  \forall c \in \mathcal C\setminus D\; \exists x \in A\cup S\cup B\; (\adj(x,c))
 \end{bmatrix*}.
\]
The first and the second rows says that $A\cup S$ and $B$ represent the sets $\sigma^{-1}(1)$ and $\sigma^{-1}(0)$ of an assignment $\sigma$, respectively.
The third row says that $A$ has no sources, and all clauses except in $D$ are satisfied by the $\sigma$ defined by $A\cup S$ and $B$.
Then, the max ones problem is equivalent to find a vertex set $Y$ maximizes the weight $\sum_{x \in Y}w'(x)$ and satisfies $G_{\Phi}\models \varphi_{\emptyset, \emptyset}(Y)$.
Here, we have
\begin{align*}
    \sat(\fg_B(G), \varphi_{S,D}) &= \begin{cases}
    \biguplus_{S'\subseteq B}\sat(G,\varphi_{S\cup S',D})\boxtimes \{S'\} & \text{if $D\cap B = \emptyset$}, \\
    \emptyset & \text{otherwise}       
    \end{cases}\\
    \sat(G\paracom H, \varphi_{S,D}) &= \biguplus_{\substack{D_1 \cap D_2 = D\\D_1, D_2 \subseteq \src(G)}}\sat(G,\varphi_{S,D_1})\boxtimes \sat(H,\varphi_{S,D_2}).
\end{align*}
Particularly, we have $p_{\max}\leq 2^{O(k)}$.

\section{Baker's Technique}\label{sec:planar}

In this section, we provide a technique for Lipschitz continuous algorithms for optimization problems on planar graphs. 
Our technique is a modification of Baker's technique~\cite{Baker94}, which provides a PTAS for many optimization problems on planar graphs. 
For simplicity, in this section, we focus on the maximum weight independent set problem on planar graphs. 
Similar techniques can be applied to obtain Lipschitz continuous algorithms for other problems, such as the minimum weight vertex cover problem on planar graphs.

Let $G=(V,E)$ be a planar graph and $w\in \mathbb{R}_{\geq 0}^V$ be a weight vector.
A vertex subset $X$ is an independent set of $G$ if no two vertices in $X$ are adjacent.
The \emph{maximum weight planar independent set} problem asks to maximize the weight $w(X):=\sum_{v\in X}w_v$ among all independent sets of $G$.
We denote this maximum value by $\OPT$.
Let $\varepsilon>0$. 
We provide a $(1-\varepsilon)$-approximation algorithm for the maximum weight planar independent set problem with Lipschitz constant $O((\varepsilon^{-1}+\log \log n)\varepsilon^{-1}\log^2 n)$.

Let $r\in V$ be an arbitrary vertex. For $i\in \mathbb{Z}_{\geq 0}$, let $R_i$ be the set of vertices at exact distance $i$ from $r$, where the distance between $r$ and $v$ is measured as the minimum number of edges in an $r$-$v$ path.
Let $m\in \mathbb{Z}_{\geq 1}$.
For $j\in \{0,\dots, m-1\}$, let $G_j$ be the subgraph of $G$ induced by $V\setminus \bigcup_{i\in \mathbb{Z}_{\geq 0}\colon i\equiv j \pmod m}R_i$.
We denote the maximum weight of an independent set of $G_j$ by $\OPT_j$.
The following lemmas are known.
\begin{lemma}[\rm{\cite{RobertsonS84}}]
For each $k\in \{0,\dots, m-1\}$, the treewidth of $G_i$ is at most $3m-2$.
\end{lemma}
\begin{lemma}[\rm{\cite{Baker94}}]
We have $\max_{j\in \{0,\dots, m-1\}}\OPT_{j}\geq \left(1-\frac{1}{m}\right)\OPT$.
\end{lemma}
In the original Baker's technique, a $(1-\varepsilon)$-approximation algorithm for the maximum weight planar independent set problem is obtained by finding the maximum weight independent set for each $G_j$ and then outputting the one with the maximum weight among them.
We make this algorithm Lipschitz continuous by replacing the max operation with a softmax operation.
We set $m=\lceil2\varepsilon^{-1}\rceil$.
Then, we take $\bm{j}^*=\softargmax^{m^{-1}}_{j\in \{0,\dots, m-1\}}\OPT_j$ and compute an approximate maximum weight independent set of $G_{\bm{j}^*}$ using the algorithm in Section~\ref{sec:DP}, where $\varepsilon$ is set as $m^{-1}$.

\begin{lemma}
Our algorithm has approximation ratio $1-\varepsilon$.
\end{lemma}
\begin{proof}
From Lemma~\ref{lem:softmax_approx}, we have $\E\left[\OPT_{\bm{j}^*}\right]\geq \left(1-\frac{1}{m}\right)\max_{j\in \{0,\dots, m-1\}}\OPT_j$.
For each $j\in \{0,\dots, m-1\}$, our framework outputs $\left(1-\frac{1}{m}\right)$-approximate solution for the maximum weight independent set problem on $G_j$. 
Therefore, the overall approximation ratio is $\left(1-\frac{1}{m}\right)^2\geq 1-\frac{2}{m}\geq 1-\varepsilon$.
% \ynote{don't we want to lower bound this?}\snote{Fixed}
\end{proof}

%\begin{lemma}
%We have $\TV\left(\bm{c},\bm{c}^{\delta u}\right)\leq \frac{4\delta}{\OPT}$.
%\end{lemma}
%\begin{proof}
%We have
%\begin{align*}
%    \TV\left(\bm{c},\bm{c}^{\delta u}\right)
%    &= \Pr\left[\bm{c}^{\delta u}\in \left[\frac{16m\log m}{\max_{k\in \{0,\dots, m-1\}}\OPT_k},\frac{16m\log m}{\max_{k\in \{0,\dots, m-1\}}\OPT^{\delta u}_k}\right]\right]\\
%    &= 2\left(1-\frac{\max_{k\in \{0,\dots, m-1\}}\OPT_k}{\max_{k\in \{0,\dots, m-1\}}\OPT^{\delta u}_k}\right)\\
%    &\leq 2\left(1-\frac{\max_{k\in \{0,\dots, m-1\}}\OPT_k}{\max_{k\in \{0,\dots, m-1\}}\OPT_k+\delta}\right)
%    \leq \frac{2\delta}{\max_{k\in \{0,\dots, m-1\}}\OPT_k}
%    \leq \frac{4\delta}{\OPT},
%\end{align*}
%where the last inequality is from $m\geq 2$.
%\end{proof}

%\begin{lemma}
%We have $\TV\left(\bm{k}^*, \bm{k}^{* \delta u}\right)\leq \frac{64\delta m\log 2m}{\OPT}$.
%\end{lemma}
%\begin{proof}
%We have
%\begin{align*}
%    \TV\left(\bm{k}^*, \bm{k}^{* \delta u}\right)
%    \leq \exp{\bm{c}\delta} - 1
%    \leq 2\bm{c}\delta 
%    \leq \frac{32\delta m\log 2m}{\max_{k\in \{0,\dots, m-1\}}\OPT_{k}}
%    \leq \frac{64\delta m\log 2m}{\OPT},
%\end{align*}
%where the second inequality is from the fact that $\exp(x)-1\leq 2x$ holds for $x\leq \frac{5}{4}$, and the last inequality is from $m\geq 2$.
%\end{proof}

\begin{lemma}
Our algorithm has a Lipschitz constant $O((\varepsilon^{-1}+\log \log n)\varepsilon^{-1}\log^2 n)$.
\end{lemma}
\begin{proof}
Let $w\in \mathbb{R}_{\geq 0}^{V}$, $u\in V$, and $\delta > 0$.
We bound $\EM\left(\bm{X}_w,\bm{X}_{w+\delta \bm{1}_u}\right)$, where $\bm{X}_w$ and $\bm{X}_{w+\delta \bm{1}_u}$ are the output of our algorithm on weight vector $w$ and $w+\delta \bm{1}_{u}$, respectively.
We denote the variable $\bm{j}^*$ corresponding to the weight vector $w$ and $w+\delta \bm{1}_u$ by $\bm{j}^*$ and $\bm{j}^*_{w+\delta \bm{1}_u}$, respectively.
From Lemma~\ref{lem:softmax_tv}, we have 
\begin{align*}
\TV\left(\bm{j}^*_w,\bm{j}^{*}_{w+\delta \bm{1}_u}\right)
\leq \frac{10m\log (2m\varepsilon^{-1})\delta}{\max_{j\in \{0,\dots, m-1\}}\OPT_j}
\leq \frac{10m\log (2p\varepsilon^{-1})\delta}{\left(1-\frac{1}{m}\right)\OPT}.
\end{align*}
Therefore, we have
\begin{align*}
\EM\left(\bm{X}_w,\bm{X}_{w+\delta \bm{1}_u}\right)
&\leq \TV\left(\bm{j}^*_w,\bm{j}^{*}_{w+\delta\bm{1}_u}\right) \OPT 
+ O\left(m+\log m + \log \log n\right)m\log^2 n\delta\\
&\leq \left(O(m\log m)+O\left(\left(m + \log \log n\right)m\log^2 n\right)\right)\delta\\
&\leq O\left(\left(m + \log \log n\right)m\log^2 n\right)\delta\\
&= O\left(\left(\varepsilon^{-1} + \log \log n\right)\varepsilon^{-1}\log^2 n\right)\delta,
\qedhere
\end{align*}
where the first inequality is from Theorem~\ref{thm:alg_independentset}.
\end{proof}

%例の論文：「重い辺を消して」「縮約して」「木幅を bounded にして DP」
%重い辺を消すところが Lipschitz 的にはヤバい。
%重い辺を消さないと辺重みの合計が OPT の定数倍で抑えられないのでそれはそれで解析が動かない。

\section{Lower Bounds for Max Ones}\label{sec:lower_bound}
% \snote{section title: max cut?}

In this section, we prove \Cref{thm:lowerbound}.
To this end, we use a known sensitivity lower bound for the maximum cut problem.
In the maximum cut problem, given a graph $G=(V,E)$, the goal is to find a bipartition of $V$ that maximizes the number of edges between the two parts.
% Here, we assume that the output of the algorithm is an assignment $\sigma:V \to \{0,1\}$, where $\sigma^{-1}(0)$ and $\sigma^{-1}(1)$ form a bipartition of $V$.
Then, the \emph{sensitivity} of a randomized algorithm $\mathcal A$ for the maximum cut problem on a graph $G=(V,E)$ is defined to be
\[
    \max_{e \in E}\EM(\mathcal A(G), \mathcal A(G-e)),
\]
where $G-e$ denotes the graph obtained from $G$ by deleting $e$ and the distance $d$ underlying the earth mover's distance is the Hamming distance, i.e., $d(S,S') = |S\triangle S'|$.
We use the following lower bound:
\begin{lemma}[\cite{fleming2024sensitivity}]\label{lem:max-cut-lb}
    There exist $\varepsilon,\delta>0$ such that any $(1-\varepsilon)$-approximation algorithm for the maximum cut problem sensitivity $\Omega(n^\delta)$, even on constant-degree graphs.
\end{lemma}

% An instance of the 3SAT problem is a pair of a weighted 3-uniform hypergraph $G=(V,E,w)$ and a set of relations $\{R_e\}_{e \in E}$, where $w:V\to \mathbb R_{\geq 0}$ and $R_e \subseteq \{0,1\}^3$.
% We say that an assignment $\varphi:V \to \{0,1\}$ satisfies the constraint on $e=(u,v,w) \in E$ if $(\varphi(u),\varphi(v),\varphi(w)) \in R_e$.
% The goal of the maximum 3SAT problem is to find an assignment that satisfies 
% In the maximum 3SAT problem, given a weighted hypergraph $G=(V,E,w)$ and a constraints $C$,

% \begin{theorem}\label{thm:max-cut-lb}
%     There exist $\varepsilon,\delta>0$ such that any $(1-\varepsilon)$-approximation algorithm for the maximum 3SAT problem has Lipschitz constant $\Omega(n^{\delta})$.
% \end{theorem}
% This implies that to achieve low Lipschitz constant, we need further assumptionas on the input.

\begin{proof}[Proof of \Cref{thm:lowerbound}]
    Let $\varepsilon,\delta>0$ be as in \Cref{lem:max-cut-lb}.
    Assume that there exists a $(1-\varepsilon)$-approximation algorithm $\mathcal A_{\rm MO}$ for the max ones problem with Lipschitz constant $s(n)$, where $n$ is the number of variables in the input 3CNF formula.
    Below, we design a low sensitivity algorithm $\mathcal A_{\rm MC}$ for the maximum cut problem using $\mathcal A_{\rm MO}$.

    Given an unweighted graph $G=(V,E)$, we construct a 3CNF formula $\varphi$ over $X := \{x_v : v \in V\} \cup \{x_{uv} : \{u,v\} \in  \binom{V}{2}\}$ as follows.
    For each pair $\{u,v\} \in \binom{V}{2}$, we add constraints over $x_u$, $x_v$, and $x_{uv}$ that are satisfied if and only if (i) $x_{uv}=1$ and $x_u \neq x_v$ or (ii) $x_{uv}=0$.
    More specifically, we add constraints $(\bar x_{uv} \vee x_u \vee x_v)$ and $(\bar x_{uv}\vee \bar x_u \vee \bar x_v)$.
    % \snote{is this 3CNF formula?} \ynote{Yes: $(\bar x_{uv} \vee x_u \vee x_v)\wedge(\bar x_{uv}\vee \bar x_u \vee \bar x_v)$}\snote{understood, but it may be better to write down}
    Note that $\varphi$ does not depend on $E$.
    Next, we define a weight function $w:X \to \mathbb R_{\geq 0}$ such that $w(x_u) = 0$ for every $u \in V$ and $w(x_{uv}) = 1$ if and only if $\{u,v\} \in E$.
    Then, we run the algorithm $\mathcal A_{\rm MO}$ on the instance $(\varphi,w)$, and let $\sigma:X \to \{0,1\}$ be the obtained solution.
    Finally, let $S := \sigma^{-1}(1) \cap V$ be the solution for $G$, and outputs $S$ with probability half and $V\setminus S$ with the remaining probability of half.

    Note that the total weight of edges cut by $S$ in $G$ is equal to the total weight of $\sigma^{-1}(1)$.
    Hence, $\mathcal A_{\rm MC}$ is a $(1-\varepsilon)$-approximation algorithm for the maximum cut problem.
    
    Now, we analyze the sensitivity of $\mathcal A_{\rm MC}$.
    Let $G=(V,E)$ and $G'=(V,E')$ be two graphs with $|E \triangle E'|=1$.
    Let $w$ and $w'$ be the weight functions constructed from $G$ and $G'$, respectively.
    Note that $\|w-w'\|_1=1$.  
    Let $\sigma$ and $\sigma'$ be the outputs of $\mathcal A_{\rm MO}$ on $(\varphi,w)$ and $(\varphi,w')$, respectively.
    Let $F = \{\{u,v\} \in E \mid \sigma(x_{uv})=1 \}$ and $F' = \{\{u,v\} \in E' \mid \sigma'(x_{uv})=1 \}$.
    Let $S$ and $S'$ be the outputs of $\mathcal A_{\mathrm MC}$ for $G$ and $G'$, respectively.
    Then, we have
    \[
        \EM(S, S') \leq \Delta \cdot \EM(F,F') = O(\EM(F,F')),
    \]
    where $\Delta = O(1)$ is the degree bound of the graphs $G$ and $G'$.
    Because $\mathcal A_{\rm MO}$ has Lipschitz constant $s(n)$ and the number of variables in $\varphi$ is $n^3$, we have 
    \begin{align*}
        & \EM(S, S') \leq O(\EM(F,F')) 
        = 
        O(\EM((\sigma^{-1}(1),w), ((\sigma')^{-1}(1),w'))) \\
        & = O\left(\frac{\EM((\sigma,w), (\sigma',w'))}{\|w-w'\|_1}\right) = O(s(n^3)).        
    \end{align*}
    This implies that $\mathcal A_{\rm MC}$ has sensitivity at most $O(s(n^3))$.
    By \Cref{lem:max-cut-lb}, we must have $s(n) = \Omega(n^{\delta/3})$.
    %\snote{first equality should be $\leq$} \ynote{which one?}\snote{$\EM(S,S')\leq \EM(\sigma^{-1}(1),(\sigma')^{-1}(1))$, because it ignores elements not in $V$}
\end{proof}
% \snote{出力が頂点なのに辺に重みがついているのは今回の状況とはちょっと違う気がします}

\section{Clique-Width}\label{sec:DP_clique}
%In this section, we prove Theorems~\ref{thm:alg_max},~\ref{thm:alg_min},~and~\ref{thm:alg_independentset}. Since the proofs are similar to that for Theorem~\ref{thm:alg_max}, we provide most of the proof of Theorem~\ref{thm:alg_min} in Appendix~\ref{app:minimization}.
%In Section~\ref{sec:dp_definition}, we define the notations and give a brief overview of the algorithm.
%Section~\ref{sec:dp_intro} handles the base case, where the graph consists of a single vertex.
%Sections~\ref{sec:dp_paracom} and~\ref{sec:dp_forget} analyze the impact of the parallel composition and the forget operations, respectively, on approximability and the Lipschitz constant.
%In Section~\ref{sec:dp_combine}, we put everything together to prove Theorem~\ref{thm:alg_max}.
%Finally, in Section~\ref{sec:dp_independent}, we provide a more refined analysis for the special case of the maximum weight independent set problem to prove Theorem~\ref{thm:alg_independentset}.

In this section, we prove \Cref{thm:alg_max_clique,thm:alg_min_clique}.
Since the proofs are similar to that for \Cref{thm:alg_max_clique}, we provide most of the proof of \Cref{thm:alg_min_clique} in \Cref{app:minimization}.
Let $G$ be a vertex-colored graph with $n$ vertices.
A \emph{VR-term} $\mathcal{T}$ of $G$ represents a procedure that constructs $G$ by recursively applying the following operations~\cite{CourcelleE2012}:
\begin{itemize}
    \item Create a vertex with color $i$.
    \item Take the disjoint union of two vertex-colored graphs with VR-terms $T_1, T_2$.
    \item Choose two different colors $i, j$ and add edges between all vertices with color $i$ and all vertices with color $j$.
    \item Choose two different colors $i, j$ and recolor all vertices with color $i$ by color $j$.
\end{itemize}
The \emph{width} of a VR-term is the number of distinct colors that appear in the overall procedure.
For a vertex $v$ in $G$, the \emph{height} of $v$ with respect to a VR-term $\mathcal{T}$ is the number of times that the disjoint union operation is applied to a graph containing $v$ during the construction procedure.  
The \emph{height} of $\mathcal{T}$ is the maximum height among all vertices with respect to $\mathcal{T}$.  
The \emph{clique-width} of a graph $G$ is the minimum width among all VR-terms of $G$.
We have the following.
\begin{lemma}[\cite{courcelle2007graph,Oum08}]\label{lem:cw_boundedheight}
Let $G$ be a graph with $n$ vertices and $\cw$ be the clique-width of $G$. 
Then, a VR-term of $G$ with width $O^*(2^\cw)$ and height $O(\log n)$ can be computed in $2^{O(\cw)} n^3$ time.
\end{lemma}
%\snote{What is time complexity?}

% \tatsuya{toriaezu kokoni}

\begin{theorem}[\cite{CourcelleE2012}]\label{thm:mso-vr-split}
    Let $r, q, k$ be positive integers.
    Let $\varphi({X})$ be a \CMSO{r}{q}-formula over $k$-colored graph with a second-order free variable ${X}$. Then, the following holds.
    \begin{enumerate}
    \item  For any $i,j \in [k]$, there exist \CMSO{r}{q}-formulas $\psi({X})$ and $\psi'({X})$ 
    such that for any colored graph $G$, we have 
    \begin{gather*}
    \sat(\eta_{i,j}(G), \varphi)  =  \sat(G, \psi), \quad \sat(\rho_{i\to j}(G), \varphi) =  \sat(G, \psi'),
    \end{gather*}
    where, $\eta_{i,j}(G)$ is the graph obtained from $G$ by adding edges between all vertices with color $i$ and all vertices with color $j$, 
    $\rho_{i\to j}(G)$ is the graph that obtained from $G$ by recoloring all vertices with color $i$ by color $j$.
    
    \item There exists a family of tuples
      $\braces{\theta_i({X}), \psi_i({X})}_{i \in [p]}$
      of \CMSO{r}{q}-formulas with a free variable ${X}$
      such that, for any $k$-colored $G$ and $H$,
    \[
        \sat(G\oplus H, \varphi) =  \biguplus_{i \in [p]} 
         \sat(G, \theta_i) \boxtimes  \sat(H, \psi_i).
    \]
    \end{enumerate}
\end{theorem}

\subsection{Definitions}\label{sec:dp_definition_clique}
%\snote{kakikake (almost same as treewidth)}
%\tatsuya{I think analysis is simpler than treewidth, because, in clique-width, we can not ``merge'' vertices like the join operation on nice tree decomposition. That is, we do not need a theorem like \cref{cor:separeted-forget-paracom}}

For a graph $G$, a weight vector $w$, and an \MSOo formula $\varphi$, we define  
\begin{align}
    \OPT_w[G,\varphi] = \max_{G\models \varphi(S)}w(S).\label{eq:def_opt_max_clique}
\end{align}
for maximization problems and
\begin{align}
    \OPT_w[G,\varphi] = \min_{G\models \varphi(S)}w(S).\label{eq:def_opt_min_clique}
\end{align}
for minimization problems.
If $\sat(G,\varphi)=\emptyset$, we define $\OPT_w[G,\varphi]=\nil$.
For graphs $G$ and \MSOo formulas $\varphi$ with $\sat(G,\varphi)\neq \emptyset$, our algorithm recursively computes the vertex set $\DP_w[G,\varphi]$ that approximately achieves the maximum in Equation~(\ref{eq:def_opt_max_clique}) or minimum in Equation~(\ref{eq:def_opt_min_clique}). 
%If no such set exists, we define $\DP_w[G,\varphi]=\nil$.
When it is clear from the context, we omit the subscript $w$ and simply write $\OPT[G,\varphi]$ and $\DP[G,\varphi]$.
We perform the dynamic programming algorithm using the formulas defined in~\Cref{thm:mso-vr-split} over the parse tree of a term obtained from \Cref{lem:cw_boundedheight}. %\snote{言葉遣いが怪しいか？}
Specifically, at every stage of the algorithm, it is guaranteed that each $X\in \sat(G,\varphi)$ satisfies $X\subseteq V(G)$.
Furthermore, our algorithm is randomized. Thus, $\DP[G,\varphi]$ can be considered as a probability distribution over vertex sets in $\sat(G,\varphi)$.

To bound the approximation ratio, for a weight vector $w\in \mathbb{R}_{\geq 0}^V$, we bound the ratio between $\E\left[w(\DP[G,\varphi])\right]$ and $\OPT[G,\varphi]$.
To bound the Lipschitz constant, for a weight vector $w\in \mathbb{R}_{\geq 0}^V$, $u\in V$, and $\delta > 0$, we bound $\EM\left(\DP_w[G,\varphi],\DP_{w+\delta \bm{1}_u}[G,\varphi]\right)$.
%\begin{align*}
%    \EM\left(\DP_w[G,\varphi],\DP_{w+\delta \bm{1}_u}[G,\varphi]\right).
%\end{align*}
From now on, we will concentrate on maximizing problems.
The algorithm for minimization problems is similar to that for maximization problems, while the detail of the analysis is slightly different.
We discuss the minimization version in \Cref{app:minimization}.

%\snote{$\nil$ に対して何かを定義しようとすると自動的に無視されますみたいなことを書く。$\OPT[G,\varphi]=\nil$ のところに対してはそもそも呼ばない、みたいな感じにすると楽か？}

\subsection{Base Case}\label{sec:dp_intro_clique}
Here we consider the base case.
Let $G$ be a graph with a single vertex $v$ and no edges, and $\varphi$ be an \MSOt formula such that $\sat(G,\varphi)$ is nonempty.
Then, $\sat(G,\varphi)$ is either $\{\emptyset\}$, $\{\{v\}\}$, or $\{\emptyset, \{v\}\}$.  
Our algorithm sets $\DP[G,\varphi] := \emptyset$ if $\sat(G,\varphi) = \{\emptyset\}$ and $\DP[G,\varphi] := \{v\}$ otherwise.  
Then, we have $w(\DP[G,\varphi]) = \OPT(G,\varphi)$, so the approximation ratio is $1$.  
Furthermore, the value of $\EM\left(\DP_w[G,\varphi],\DP_{w+\delta \bm{1}_u}[G,\varphi]\right)$ is $\delta$ if $u=v$ and $0$ otherwise.
%\ynote{what do we do when $v$ appears in  $\sat(G,\varphi)$?}
%\snote{we (want to) ensure $v$ does not appear in $\sat(G,\varphi)$. This is why the modification of DP is needed.}

\subsection{Disjoint Union}\label{sec:dp_disjoint_clique}

Next, we consider the disjoint union.
The algorithm and discussion here are the same as those in \Cref{sec:dp_paracom}, where the parallel composition is replaced by the disjoint union.
Therefore, we omit the technical details and only provide the definition of the DP table, as well as the bounds on approximability and Lipschitz continuity.  
Let
\begin{align}
    \sat(G\oplus H, \varphi) =  \biguplus_{i\in [p]}\sat(G, \theta_{i})  \boxtimes \sat(H, \psi_{i}).\label{eq:dp_paracom_clique}
\end{align}
Assume we have already computed $\DP[G, \theta_{i}]$ and $\DP[H, \psi_{i}]$ for each $i\in [p]$.
%We compute $\DP[G\paracom H, \varphi]$.
For each $i\in [p]$, we denote $\OPT_{i}:=\OPT[G, \theta_{i}] + \OPT[H, \psi_{i}]$.
%\begin{align*}
%    \OPT_{i}:=\OPT[G, \theta_{i}] + \OPT[H, \psi_{i}].
%\end{align*}
%We first sample $\bm{c}\in \mathbb{R}_{>0}$ uniformly from $\left[\frac{2\log (2p\varepsilon^{-1})}{\varepsilon\OPT[G\paracom H,\varphi,S]},\frac{4\log (2p\varepsilon^{-1})}{\varepsilon\OPT[G\paracom H,\varphi,S]}\right]$.
%Let $i\in [p]$. and $j\in [q_i]$, we denote 
%\begin{align*}
%\OPT_i=\OPT[G, \theta_{i},S\cap V(G)] + \OPT[H, \psi_{i},S\cap V(H)].
%\end{align*}
%By definition, we have $\max_{i\in [p]}\OPT_i = \OPT[G\paracom H, \varphi]$.
Then, we take $\bm{i}^*=\softargmax^{\varepsilon}_{i\in [p]}\OPT_i$ and define
$\DP[G\oplus H, \varphi] :=  \DP[G, \theta_{\bm{i}^*}] \cup \DP[H, \psi_{\bm{i}^*}]$.
%\begin{align*}
%    \DP[G\paracom H, \varphi] :=  \DP[G, \theta_{\bm{i}^*}] \cup \DP[H, \psi_{\bm{i}^*}].
%\end{align*}
%We define the probability distribution $\bm{i}^*$ over $[p]$ so that 
%\begin{align*}
%\Pr[\bm{i}^*=i]=\frac{\exp\left(\bm{c}\OPT_i\right)}{\sum_{i'=1}^{p}\exp\left(\bm{c}\OPT_i\right)}.
%\end{align*}
%Then we define 
%\begin{align*}
%    \DP[G\paracom H, \varphi,S] :=  \DP[G, \theta_{\bm{i}^*},S\cap V(G)] \cup \DP[H, \psi_{\bm{i}^*},S\cap V(H)].
%\end{align*}
%First we analyze the approximation ratio.
We have the following lemmas.
\begin{lemma}\label{lem:paracom_approx_clique}
Let $0 < \alpha \leq 1$ and suppose $\E\left[w(\DP[G,\theta_i])\right]\geq \alpha \OPT[G,\theta_i]$ and $\E\left[w(\DP[H,\psi_i])\right]\geq \alpha \OPT[H,\psi_i]$
%\begin{align*}
%    \E\left[w(\DP[G,\theta_i])\right]\geq \alpha \OPT[G,\theta_i] \quad \text{and} \quad 
%    \E\left[w(\DP[H,\psi_i])\right]\geq \alpha \OPT[H,\psi_i]
%\end{align*}
hold for all $i\in [p]$.
Then, we have $\E\left[w\left(\DP[G\oplus H,\varphi]\right)\right]\geq (1-\varepsilon)\alpha\OPT[G\oplus H,\varphi]$.
\end{lemma}
\begin{lemma}\label{lem:paracom_lip_clique}
%Let $c\in \mathbb{R}_{\geq 0}$ and assume $\bm{c}=\bm{c}^{\delta u}=c$. 
Let $\beta\in \mathbb{R}_{\geq 0}$ and suppose $\EM\left(\DP_{w}[G,\theta_i],\DP_{w+\delta\bm{1}_{u}}[G,\theta_i]\right)\leq \beta$
%\begin{align*}
%    \EM\left(\DP_{w}[G,\theta_i],\DP_{w+\delta\bm{1}_{u}}[G,\theta_i]\right)\leq \beta
%\end{align*}
holds for all $i\in [p]$.
Then, we have $\EM\left(\DP_{w}[G\oplus H,\varphi],\DP_{w+\delta\bm{1}_{u}}[G\oplus H,\varphi]\right) \leq 30\varepsilon^{-1}\log(2p\varepsilon^{-1})\delta + \beta$.
%\begin{align*}
%    \EM\left(\DP_{w}[G\paracom H,\varphi],\DP_{w+\delta\bm{1}_{u}}[G\paracom H,\varphi]\right) \leq 30\varepsilon^{-1}\log(2p\varepsilon^{-1})\delta + \beta.
%\end{align*}
\end{lemma}

\subsection{Putting Together}

%\snote{atode}
% \snote{wrote. Are those words correct?}
Here we complete the proof of \Cref{thm:alg_max_clique}.
Let $G=(V,E)$ be a graph with clique-width $\cw$.
From \Cref{lem:td_boundeddiam}, we can compute a VR-term $\mathcal{T}$ over a $k$-colored graph denoting $G$ with height $O(\log n)$.
We perform the dynamic programming algorithm on the parse tree of $\mathcal{T}$.
We remark here that operations other than disjoint union do not affect either the approximability or the Lipschitz continuity, by setting  
\begin{align*}
    \DP[\eta_{i,j}(G), \varphi]  :=  \DP[G, \psi], \quad \DP[\rho_{i\to j}(G), \varphi] :=  \DP[G, \psi'].
\end{align*}
We have the following.
\begin{lemma}\label{lem:combineall_clique}
Let $\varepsilon \in (0,1]$ and $h$ be the height of the VR-term $\mathcal{T}$.
Our algorithm outputs a solution $X\in \sat(G,\varphi)$ such that $\E[w(X)]\geq (1-h\varepsilon)\OPT[G,\varphi]$.
The Lipschitz constant is bounded by $30(h+1)\varepsilon^{-1}\log(2p_{\max}\varepsilon^{-1})$, where $p_{\max}$ is the maximum of $p$ among all update formula~(\ref{eq:dp_paracom_clique}) that the algorithm uses.
\end{lemma}
\begin{proof}
The approximability bound is obtained by repeatedly applying \Cref{lem:paracom_approx_clique} and $(1-\varepsilon)^h\geq 1-h\varepsilon$.
The Lipschitzness bound is obtained by repeatedly applying \Cref{lem:paracom_lip_clique}.
\end{proof}

Since $p_{\max}$ is bounded by a function of $k$ and $|\varphi|$, we have the following.
%\begin{theorem}
%Let $G$ be a graph with treewidth $t$, $\varphi$ be an \MSOt formula, and $\varepsilon \in \mathbb{R}_{>0}$.
%Then, there is a $(1-\varepsilon)$-approximation algorithm for \MSOt maximization with Lipschitz constant $(f(t,|\varphi|)+\log \varepsilon^{-1}+\log \log n)\varepsilon^{-1}\log^2 n$ that runs in FPT time parameterized by $t$ and $|\varphi|$.
%\end{theorem}
\begin{proof}[Proof of \Cref{thm:alg_max_clique}]
    The claim follows by substituting $\varepsilon$, $h$, and $k$ in \Cref{lem:combineall} with $\frac{\varepsilon}{\Theta(\log n)}$, $O(\log n)$, and $O^*(2^k)$, respectively.
\end{proof}

\bibliography{bibliography}

\appendix
\section{Omitted Proofs on Logic}\label{app:logics}
\begin{proof}[Proof of \cref{prop:height-equivalence}]
Let $(\mathcal B, \mathcal T)$ be a tree decomposition of $G$.
Let $t\in V(\mathcal T)$.
Let $\mathcal T_t$ denote the subtree of $\mathcal T$ consisting of all descendants of $t$.
Note that $\mathcal T_t \subseteq V(G)$ since a node of $\mathcal T$ is in $\mathcal B$.

	We recursively transform a tree decomposition $(\mathcal B, \mathcal T)$ to a \HR-parse tree.
	In this proof, we denote the height of the parse tree of a term $t$ by $\height(t)$.
	Let $v=B_v$ be a node of $\mathcal T$. We have the following three cases.
	\begin{description}
		\item[$v$ is a leaf node:] 
    For a graph $H$, let $\Cons{Graph}(H)$ denote the term $(\paracom_{i\in V(H)} \Cons{i}) \paracom (\paracom_{\{i,j\} \in E(H)}\Cons{ij})$.
    We arbitrarily identify $B_v$ with $\{1, \ldots, |B_v|\}$.
		Then, it is clear that the term $\Cons{Graph}(G[B_v])$ denote the graph $G[B_v]$.
		When $\Cons{Graph}(G[B_v])$ is expanded linearly, its height is $O(|B_v|^2) = O(k^2)$.
		However, by restructuring it to a balanced binary tree, the height can be reduced to $O(\log k)$ since parallel-composition operation has the associative property.
		% We write term $t_v$ representing a graph $H = G[B_v]$ by $\Cons{Graph}(H)$. \ynote{I couldn't understand this sentence}
		\item[$v$ has exactly one child node $c=B_c$:] 
		Let $t_c$ be the term representing $G[\mathcal T_c]$.
		Let $G_\Delta = G[B_v \setminus B_c]$.
		Then, the term $\Cons{Graph}(G_\Delta) \paracom \fg_{B_c \setminus B_v}(t_c)$ denotes the graph $G[\mathcal T_v]$ and its height is at most $\max\{\log(k),\height(t_c)\}+2$.

		\item[$v$ has exactly two child nodes $c_1=B_{c_1}$ and $c_2=B_{c_2}$:] 
		Let $t_1$ and $t_2$ be the terms representing $G[\mathcal T_{c_1}]$ and $G[\mathcal T_{c_2}]$, respectively.
		Let $G_\Delta = G[B_v \setminus (B_{c_1} \cup B_{c_2})]$.
		Then, the term $\Cons{Graph}(G_\Delta) \paracom \fg_{B_{c} \setminus B_v}(t_1) \paracom \fg_{B_{c_2} \setminus B_v}(t_{2})$
		denotes the graph $G[\mathcal T_v]$ and the height is at most $\max\{\log(k),\allowbreak \height\parens{t_1}, \height\parens{t_2}\}+3$.
	\end{description}
	It is easy to verify that the resulting parse tree represents $G$, and the height is at most $O(\log k + h)$. %\snote{$w$ is $k$?}
\end{proof}

\begin{proof}[Proof of \cref{cor:separeted-forget-paracom}]
    First we show the first statement of \cref{cor:separeted-forget-paracom}.
    Since $S' \subseteq B \subseteq \src(G)$, $\sat(G,\psi\frest{S\cup S'})$ and $\{S'\}$ is separated.
    It is clear that families of vertex sets $(\sat(G,\psi\frest{S\cup S'})\boxtimes \{S'\})_{S' \subseteq B}$ are disjoint.
    
    Let $\psi(X)$ be a \CMSO{r}{q}-formula obtained by applying the first statement of \cref{thm:mso-fg-split} to $\varphi(X)$ and $B$.
    Then, $\fg_B(G) \models \varphi(A) \iff G \models \psi(A)$ for any vertex set $A$.
    We show that $\sat(\fg_B(G), \varphi\frest{S}) = \biguplus_{S'\subseteq B}\sat(G,\psi\frest{S\cup S'})\boxtimes \{S'\}.$
    
    Let $A \in \sat(\fg_B(G), \varphi\frest{S})$ and $A' = (\src(G) \cap A)\setminus S$.
    Since $\fg_B(G) \models \varphi \frest{S}(X)$, we have 
    $\fg_B(G) \models \varphi(A \cup S)$ and $\fg_B(G) \models \nosrc(A)$.
    Thus, $G \models \psi(A\cup S)$ and then $G \models \psi((A \setminus A') \cup (A \cup A'))$.
    Furthermore, $A \cap \src({\fg_B(G)}) = \varnothing$ and thus,
    \[
      A'= (\src(G) \cap A) \setminus S 
      = \parens[big]{\parens{\src(G) \setminus \src(\fg_B(G))} \cap A}\setminus S 
      = (B \cap A) \setminus S.
    \]
    Since $B \cap S = \varnothing$, we have $A' = B \cap A$, and thus, $G \models \nosrc(A\setminus A')$.
    Thus, $A \in \sat(G,\psi\frest{S\cup A'})\boxtimes \{A'\}$ and $A' \subseteq B$.
    
    Let $A \in \sat(G,\psi\frest{S\cup S'})\boxtimes \{S'\}$ for some $S' \subseteq B$.
    Then, $G \models \psi((A\setminus S') \cup S \cup S')$.
    Thus, $\fg_B(G) \models \varphi(A \cup S)$.
    Since $G \models \nosrc(A\setminus S')$ and $\src(\fg_B(G)) = \src(G) \setminus B \subseteq \src(G) \setminus S'$,
    we have $A \cap \src(\fg_B(G)) = \varnothing$. Thus, $A \in \sat(\fg_B(G), \varphi\frest{S})$.
    
    Next, we show the second statement of \cref{cor:separeted-forget-paracom}.
    From \cref{thm:mso-fg-split}, 
      there exists a family of tuples $\braces{\theta'_i({X}), \psi'_i({X})}_{i\in [p']}$
        of \CMSO{r}{q}-formulas with a free variable ${X}$
        such that, for any \kgraphs{$k$} $G$ and $H$,
        \begin{equation*} % \label{eq:inpr:sepparacom:paracom}
            \sat(G\paracom H, \varphi) =  \biguplus_{ i \in [p']} 
             \{S \cup P : S \in \sat(G, \theta'_i),  P \in \sat(H, \psi'_i)\}.
        \end{equation*}
    Then, we have
    \begin{align*}
        \sat&(G\paracom H, \varphi\frest{S}) \boxtimes \{S\}\\
        &= \sat(G\paracom H, \varphi)\cap \{X\colon \src\parens{G\paracom H}\cap X = S\}\\
        & 
        =\parens[\Bigg]{
          \biguplus_{i\in[p]}
             \braces*{S \cup P : S \in \sat(G, \theta'_i),  P \in \sat(H, \psi'_i)}
          } 
          \cap \{X\colon \src\parens{G\paracom H}\cap X = S\}\\
        & \begin{aligned}
          = \biguplus_{i\in [p]}\bigcup_{S'\subseteq \src\parens{G}}\bigcup_{S''\subseteq \src(H)}
          \parens{\{S' \cup S''\} \boxtimes \sat(G,\theta_i\frest{S'})\boxtimes 
            \sat(H,\psi_i\frest{S''})} \\ 
           {} \cap \{X\colon \src\parens{G\paracom H}\cap X = S\}
        \end{aligned}\\
        &= \biguplus_{i\in [p]}\bigcup_{\substack{(S', S'')\colon\\ S'\cup S'' = S,\\S'\subseteq \src(G),\ S''\subseteq \src(H)}}\sat(G,\theta_i\frest{S'})\boxtimes \sat(H,\psi_i\frest{S''})\boxtimes \{S\}.
    \end{align*}
    Thus, there exists a family of tuples
    $\braces{\theta_i\frest{S_i}({X}), \psi_i\frest{S'_i}({X})}_{i \in [p]}$ of \CMSO{r}{q}-formulas with free variables ${X}$
    such that, for any \kgraphs{$k$} $G$ and $H$,
    \[
        \sat(G\paracom H, \varphi\frest{S}) =
        \bigcup_{i \in [p']}\sat(G,\theta_i\frest{S_i})\boxtimes \sat(H,\psi_i\frest{S'_i}).
    \]
    Using the same method as the proof of \cite[Propopsition 5.37]{CourcelleE2012}, $\bigcup$ can be replaced with the disjoint union $\biguplus$,
    and the proof is complete.
\end{proof}
\section{Analysis of Softmax and Softmin}\label{app:softmax}

In this section, we provide the approximation ratios and stabilities of $\softmax$ and $\softmin$.
The following lemma is useful.
\begin{lemma}
Let $x\geq \delta>0$ and $q > 0$.
Consider sampling $\bm{c}$ uniformly from $\left[\frac{q}{x},\frac{2q}{x}\right]$ and $\bm{c}'$ uniformly from $\left[\frac{q}{x+\delta},\frac{2q}{x+\delta}\right]$.
Then, we have $\TV(\bm{c},\bm{c}')\leq \frac{2\delta}{x}$.
\end{lemma}
\begin{proof}
We have
\[
    \TV\left(\bm{c},\bm{c}'\right)
     = \Pr\left[\bm{c}\in \left[\frac{2q}{x+\delta},\frac{2q}{x}\right]\right]
    = \frac{\frac{2q}{x}-\frac{2q}{x+\delta}}{\frac{2q}{x}-\frac{q}{x}}
    = 2 - \frac{2x}{x+\delta}
    = \frac{2\delta}{x+\delta}\leq \frac{2\delta}{x}.\qedhere
\]
\end{proof}

Now we analyze the $\softmax$ operation.

\begin{proof}[Proof of Lemma~\ref{lem:softmax_approx}]
We have
\begin{align*}
    \E\left[\softmax^{\epsilon}_{i\in [p]}x_i\right]
    &\geq \Pr\left[\softmax^{\epsilon}_{i\in [p]}x_i \geq \left(1-\frac{\epsilon}{2}\right)\max_{i\in [p]} x_i\right]\cdot \left(1-\frac{\epsilon}{2}\right)\max_{i\in [p]}x_i\\
    &\geq \left(1-p\cdot \frac{\exp\left(c\left(1-\frac{\epsilon}{2}\right)\max_{i\in[p]}x_i\right)}{\sum_{i\in [p]}\exp(cx_i)}\right)\cdot \left(1-\frac{\epsilon}{2}\right)\max_{i\in [p]}x_i\\
    &\geq \left(1-p\cdot \frac{\exp\left(c\left(1-\frac{\epsilon}{2}\right)\max_{i\in[p]}x_i\right)}{\max_{i\in [p]}\exp(cx_i)}\right)\cdot \left(1-\frac{\epsilon}{2}\right)\max_{i\in [p]}x_i\\
    &= \left(1-p\cdot \frac{\exp\left(c\left(1-\frac{\epsilon}{2}\right)\max_{i\in[p]}x_i\right)}{\exp\left(c\max_{i\in [p]}x_i\right)}\right)\cdot \left(1-\frac{\epsilon}{2}\right)\max_{i\in [p]}x_i\\
    &= \left(1-p\cdot \exp\left(-\frac{c\epsilon}{2}\max_{i\in[p]}x_i\right)\right)\cdot \left(1-\frac{\epsilon}{2}\right)\max_{i\in [p]}x_i\\
    &\geq \left(1-p\cdot \exp\left(\log(2p\epsilon^{-1})\right)\right)\cdot \left(1-\frac{\epsilon}{2}\right)\max_{i\in [p]}x_i\\
    &= \left(1-\frac{\epsilon}{2}\right)\cdot \left(1-\frac{\epsilon}{2}\right)\max_{i\in [p]}x_i
    \geq (1-\epsilon)\max_{i\in [p]}x_i.\qedhere
\end{align*}
\end{proof}

\begin{proof}[Proof of Lemma~\ref{lem:softmax_tv}]
Let $\bm{c}$ and $\bm{c}'$ be the vectors sampled when computing $\softargmax^{\epsilon}_{i\in [p]}x_i$ and $\softargmax^{\epsilon}_{i\in [p]}x_i'$, respectively.
Let $I = \left[\frac{2\epsilon^{-1}\log(2p\epsilon^{-1})}{\max_{i\in [p]}x_i},\frac{4\epsilon^{-1}\log(2p\epsilon^{-1})}{\max_{i\in [p]}x_i'}\right]$.
For each parameter $c\in I$, we consider transporting probability mass for $\bm{c}=c$ to that for $\bm{c}'=c$ as far as possible. The remaining mass is transported arbitrarily. Then, we have
\begin{align*}
    &\TV\left(\softargmax^{\epsilon}_{i\in [p]}x_i, \softargmax^{\epsilon}_{i\in [p]}x_i'\right)\\
    &\leq \TV\left(\bm{c},\bm{c}'\right) + \sup_{c\in I}\TV\left(\left(\softargmax^{\epsilon}_{i\in [p]}x_i\mid \bm{c}=c\right), \left(\softargmax^{\epsilon}_{i\in [p]}x_i'\mid \bm{c'}=c\right)\right).
\end{align*}
Now, for $c\in I$, we have
\begin{align*}
    &\TV\left(\left(\softargmax^{\epsilon}_{i\in [p]}x_i\mid \bm{c}=c\right), \left(\softargmax^{\epsilon}_{i\in [p]}x_i'\mid \bm{c'}=c\right)\right)\\
    &=\sum_{i\in [p]}\max\left(0,\Pr\left[\softargmax^{\epsilon}_{i\in [p]}x_i'=i\mid \bm{c}'=c\right]-\Pr\left[\softargmax^{\epsilon}_{i\in [p]}x_i=i\mid \bm{c}=c\right]\right)\\
    &= \sum_{i\in [p]}\max\left(0,\frac{\exp(cx_{i}')}{\sum_{i'\in [p]}\exp(cx_{i'}')}-\frac{\exp(cx_{i})}{\sum_{i'\in [p]}\exp(cx_{i'})}\right)\\
    &\leq \sum_{i\in [p]}\frac{\exp(cx_{i}')-\exp(cx_{i})}{\sum_{i'\in [p]}\exp(cx_{i'})}
    = \sum_{i\in [p]}\frac{(\exp(c(x_{i}'-x_{i}))-1)\exp(cx_{i})}{\sum_{i'}\exp(cx_{i'})}\\
    &\leq \sum_{i\in [p]}\frac{(\exp(c\delta)-1)\exp(cx_{i})}{\sum_{i'}\exp(cx_{i'})}
    = \exp(c\delta)-1
    \leq 2c\delta
    \leq \frac{8\epsilon^{-1}\log (2p\epsilon^{-1})\delta}{\max_{i\in [p]}x_i},
\end{align*}
where the fourth inequality is from the fact that $\exp(x)-1\leq 2x$ holds for $x \leq 1$.
Therefore, we have
\[
    \TV\left(\softargmax^{\epsilon}_{i\in [p]}x_i, \softargmax^{\epsilon}_{i\in [p]}x_i'\right)
    \leq \frac{2\delta}{\max_{i\in [p]}x_i} + \frac{8\epsilon^{-1}\log (2p\epsilon^{-1})\delta}{\max_{i\in [p]}x_i}\leq \frac{10\epsilon^{-1}\log (2p\epsilon^{-1})\delta}{\max_{i\in [p]}x_i}.\qedhere
\]
\end{proof}

Now we analyze the $\softmin$ operation.
\begin{proof}[Proof of Lemma~\ref{lem:softmin_approx}]
For $t\geq \frac{1}{2}$, we have
\begin{align*}
    \Pr\left[x_{\bm{i}^*}\geq \left(1+t\epsilon\right)\min_{i\in [p]}x_i\right]
    &\leq \frac{p\exp(-c(1+t\epsilon)\min_{i\in [p]}x_i)}{\sum_{i\in [p]}\exp(-cx_i)}
    \leq \frac{p\exp(-c(1+t\epsilon)\min_{i\in [p]}x_i)}{\max_{i\in [p]}\exp(-cx_i)}\\
    &= \frac{p\exp(-c(1+t\epsilon)\min_{i\in [p]}x_i)}{\exp(-c\min_{i\in [p]}x_i)}
    = p\exp\left(-ct\epsilon\min_{i\in [p]}x_i\right)\\
    &= p\exp(-tb\log(2p\epsilon^{-1}))
    \leq p\cdot \left(\frac{\epsilon}{2p}\right)^{2t} 
    \leq \left(\frac{\epsilon}{2}\right)^{2t}.
\end{align*}
Therefore, we have
\begin{align*}
    \E[x_{\bm{i}^*}]
    &= \int_{0}^{\infty}\Pr\left[x_{\bm{i}^*}\geq \left(1+t\epsilon\right)\min_{i\in [p]}x_i\right]dt\cdot  \min_{i\in [p]}x_i\\
    &\leq \left(\left(1+\frac{\epsilon}{2}\right)+\int_{1+\frac{\epsilon}{2}}^{\infty}\Pr\left[x_{\bm{i}^*}\geq \left(1+t\epsilon\right)\min_{i\in [p]}x_i\right]dt\right)\cdot  \min_{i\in [p]}x_i\\
    &\leq \left(\left(1+\frac{\epsilon}{2}\right)+\int_{1+\frac{\epsilon}{2}}^{\infty}\left(\frac{\epsilon}{2}\right)^{2t}dt\right)\cdot  \min_{i\in [p]}x_i\\
    &\leq \left(\left(1+\frac{\epsilon}{2}\right)+\left(\frac{\epsilon}{2}\right)^{2+\epsilon}\cdot 2\log\left(2\epsilon^{-1}\right)\right)\cdot  \min_{i\in [p]}x_i\\
    &\leq \left(\left(1+\frac{\epsilon}{2}\right)+\left(\frac{\epsilon}{2}\right)^{2}\cdot 2\right)\cdot  \min_{i\in [p]}x_i
    \leq (1+\epsilon)\min_{i\in [p]}x_i.\qedhere
\end{align*}
\end{proof}

\begin{proof}[Proof of Lemma~\ref{lem:softmin_tv}]
Let $\bm{c}$ and $\bm{c}'$ be the vectors sampled when computing $\softargmin^{\epsilon}_{i\in [p]}x_i$ and $\softargmin^{\epsilon}_{i\in [p]}x_i'$, respectively.
Let $I = \left[\frac{2\epsilon^{-1}\log(2p\epsilon^{-1})}{\min_{i\in [p]}x_i},\frac{4\epsilon^{-1}\log(2p\epsilon^{-1})}{\min_{i\in [p]}x_i'}\right]$.
For each parameter $c\in I$, we consider transporting probability mass for $\bm{c}=c$ to that for $\bm{c}'=c$ as far as possible. The remaining mass is transported arbitrarily. Then, we have
\begin{align*}
    &\TV\left(\softargmin^{\epsilon}_{i\in [p]}x_i, \softargmin^{\epsilon}_{i\in [p]}x_i'\right)\\
    &\leq \TV\left(\bm{c},\bm{c}'\right) + \sup_{c\in I}\TV\left(\left(\softargmin^{\epsilon}_{i\in [p]}x_i\mid \bm{c}=c\right), \left(\softargmin^{\epsilon}_{i\in [p]}x_i'\mid \bm{c'}=c\right)\right).
\end{align*}
Now, for $c\in I$, we have
\begin{align*}
    &\TV\left(\left(\softargmin^{\epsilon}_{i\in [p]}x_i\mid \bm{c}=c\right), \left(\softargmin^{\epsilon}_{i\in [p]}x_i'\mid \bm{c'}=c\right)\right)\\
    &=\sum_{i\in [p]}\max\left(0,\Pr\left[\softargmin^{\epsilon}_{i\in [p]}x_i=i\mid \bm{c}=c\right]-\Pr\left[\softargmin^{\epsilon}_{i\in [p]}x_i'=i\mid \bm{c}'=c\right]\right)\\
    &= \sum_{i\in [p]}\max\left(0,\frac{\exp(-cx_{i})}{\sum_{i'\in [p]}\exp(-cx_{i'})}-\frac{\exp(-cx_{i}')}{\sum_{i'\in [p]}\exp(-cx_{i'}')}\right)\\
    &\leq \sum_{i\in [p]}\frac{\exp(-cx_{i})-\exp(-cx_{i}')}{\sum_{i'\in [p]}\exp(-cx_{i'})}
    = \sum_{i\in [p]}\frac{(1-\exp(-c(x_{i}'-x_{i})))\exp(-cx_{i})}{\sum_{i'\in [p]}\exp(-cx_{i'})}\\
    &\leq \sum_{i\in [p]}\frac{(1-\exp(-c\delta))\exp(-cx_{i})}{\sum_{i'\in [p]}\exp(-cx_{i'})}
    = 1-\exp(-c\delta)
    \leq 2c\delta
    \leq \frac{8\epsilon^{-1}\log (2p\epsilon^{-1})\delta}{\min_{i\in [p]}x_i},
\end{align*}
where the fourth inequality is from the fact that $1-\exp(-x)\leq x\leq 2x$ holds for $x \leq 1$.
Therefore, we have
\[
    \TV\left(\softargmin^{\epsilon}_{i\in [p]}x_i, \softargmin^{\epsilon}_{i\in [p]}x_i'\right)
    \leq \frac{2\delta}{\min_{i\in [p]}x_i} + \frac{8\epsilon^{-1}\log (2p\epsilon^{-1})\delta}{\min_{i\in [p]}x_i}\leq \frac{10\epsilon^{-1}\log (2p\epsilon^{-1})\delta}{\min_{i\in [p]}x_i}.\qedhere
\]
\end{proof}
\section{Minimization}\label{app:minimization}

\subsection{Treewidth}\label{sec:minimize_tree}

Here, we provide the minimization version of the framework from Section~\ref{sec:DP}. 
The algorithm is obtained simply by replacing $\softargmax$ with $\softargmin$. 
Some parts of the analysis are identical to that in Section~\ref{sec:DP}. 
Specifically, for the introduce operation, exactly the same results can be obtained with exactly the same proof. 
Furthermore, by replacing $(1+\epsilon)$ with $(1-\epsilon)$ and reversing all inequalities in the proofs of Lemmas~\ref{lem:paracom_approx}~and~\ref{lem:forget_approx}, we obtain the following.
\begin{lemma}\label{lem:paracom_approx_min}
Let $0 < \alpha \leq 1$ and assume 
\begin{align*}
    \E\left[w(\DP[G,\varphi^G_i])\right]&\leq \alpha \OPT[G,\varphi^G_i],\\
    \E\left[w(\DP[H,\varphi^H_i])\right]&\leq \alpha \OPT[H,\varphi^H_i]
\end{align*}
holds for all $i\in [p]$.
Then, we have $\E\left[w\left(\DP[G\paracom H,\varphi]\right)\right]\leq (1+\epsilon)\alpha\OPT[G\paracom H,\varphi]$.
\end{lemma}

\begin{lemma}\label{lem:forget_approx_min}
Let $0 < \alpha \leq 1$ and assume 
\begin{align*}
    \E\left[w(\DP[G,\varphi_i])\right]&\leq \alpha \OPT[G,\varphi_i]
\end{align*}
holds for all $i\in [p]$.
Then, we have $\DP[\fg_{B}(G),\varphi]\leq (1+\epsilon)\alpha\OPT[\fg_{B}(S),\varphi]$.
\end{lemma}

In the analysis of Lipschitz continuity, $\OPT$ cannot be used as a trivial bound for the earth mover's distance. Thus, we need a slightly different proof.
\begin{lemma}\label{lem:paracom_lip_min}
Let $0 < \alpha \leq 1$, $\beta\in \mathbb{R}_{\geq 0}$ and assume 
\begin{align*}
    \E\left[w(\DP[G\paracom H,\varphi])\right]&\leq \alpha \OPT_w[G\paracom H,\varphi],\\
    \E\left[(w+\delta \bm{1}_u)(\DP[G\paracom H,\varphi])\right]&\leq \alpha \OPT_{w+\delta\bm{1}_u}[G\paracom H,\varphi],\\
    \EM\left(\DP_{w}[G,\varphi^{G}_i],\DP_{w+\delta\bm{1}_{u}}[G,\varphi^{G}_i]\right)&\leq \beta
\end{align*}
holds for all $i\in [p]$.
Then, we have
\begin{align*}
    \EM\left(\DP_{w}[G\paracom H,\varphi],\DP_{w+\delta\bm{1}_{u}}[G\paracom H,\varphi]\right) \leq 30\alpha\epsilon^{-1}\log(2p\epsilon^{-1})\delta + \beta.
\end{align*}
\end{lemma}
\begin{proof}
We have
\begin{align*}
    &\EM\left(\DP_{w}[G\paracom H,\varphi],\DP_{w+\delta\bm{1}_{u}}[G\paracom H,\varphi]\right)\\
    &\leq \TV\left(\bm{i}^{*}_w,\bm{i}^{*}_{w+\delta \bm{1}_u}\right)\alpha\left(\OPT_w[G\paracom H,\varphi] + \OPT_{w+\delta \bm{1}_u}[G\paracom H, \varphi]\right)
    + \max_{i\in [p]}\left(\EM\left(\DP_{w}[G,\varphi^{G}_i],\DP_{w+\delta \bm{1}_u}[G,\varphi^{G}_i]\right)\right)\\
    &\leq \frac{10\alpha \epsilon^{-1}\log (2p\epsilon^{-1})\delta}{\OPT_w[G\paracom H, \varphi,S]}\left(\OPT_w[G\paracom H,\varphi] + \OPT_{w+\delta \bm{1}_u}[G\paracom H, \varphi]\right) + \beta
    \leq 30\alpha \epsilon^{-1}\log(2p\epsilon^{-1})\delta+\beta.
\end{align*}
\end{proof}

\begin{lemma}\label{lem:forget_lip_min}
Let $\beta\in \mathbb{R}_{\geq 0}$ and assume
\begin{align*}
    \E\left[w(\DP[\fg_{B}(G),\varphi])\right]&\leq \alpha \OPT_w[\fg_B(G),\varphi],\\
    \E\left[(w+\delta \bm{1}_u)(\DP[\fg_B(G),\varphi])\right]&\leq \alpha \OPT_{w+\delta\bm{1}_u}[\fg_B(G),\varphi],\\
    \EM\left(\DP_{w}[G,\varphi_{i}],\DP_{w+\delta \bm{1}_u}[G,\varphi_{i}]\right)\leq \beta
\end{align*}
holds for all $i\in [p]$. Then, we have
\begin{align*}
    \EM\left(\DP_{w}[\fg_{B}(G),\varphi],\DP_{w+\delta\bm{1}_{u}}[\fg_{B}(G),\varphi]\right)\leq 31\alpha \epsilon^{-1}\log(2p\epsilon^{-1})\delta + \beta.
\end{align*}
\end{lemma}
\begin{proof}
We have
\begin{align*}
    &\EM\left(\DP_{w}[\fg_{B}(G),\varphi],\DP_{w+\delta\bm{1}_{u}}[\fg_{B}(G),\varphi]\right)\\
    &\leq \TV\left(\bm{i}^{*}_w,\bm{i}^{*}_{w+\delta \bm{1}_u}\right) \alpha \left(\OPT_{w}[\fg_{B}(G),\varphi]+\OPT_{w+\delta \bm{1}_u}[\fg_{B}(G),\varphi]\right)\\
    &\quad + \max_{i\in [p]}\left(\EM\left(\DP_{w}[G,\varphi_i]\cup S_i,\DP_{w+\delta \bm{1}_u}[G,\varphi_i]\cup S_i\right)\right)\\
    &\leq \frac{10\alpha \epsilon^{-1}\log(2p\epsilon^{-1})\delta}{\OPT_w[\fg_{B}(G),\varphi]}\left(\OPT_w[\fg_{B}(G),\varphi]+\OPT_{w+\delta \bm{1}_u}[\fg_{B}(G),\varphi]\right) \\
    &\quad + \max_{i\in [p]}\left(\EM\left(\DP_{w}[G,\varphi_i],\DP_{w+\delta \bm{1}_u}[G,\varphi_i]\right)\right) + \max_{i\in [p]}\left[d_{\mathrm{w}}((S_i,w),(S_i,w+\delta \bm{1}_u))\right]\\
    &\leq 30\alpha \epsilon^{-1}\log(2p\epsilon^{-1})\delta + \beta + \delta
    \leq 31\alpha \epsilon^{-1}\log(2p\epsilon^{-1})\delta + \beta.\qedhere
\end{align*}
\end{proof}

The following lemma corresponds to Lemma~\ref{lem:combineall}, but slightly different.
\begin{lemma}\label{lem:combineall_min}
Let $\epsilon \in \mathbb{R}_{>0}$ and $h$ be the height of $T$.
Our algorithm outputs a solution $X\in \sat(G,\varphi)$ such that $\E[w(X)]\leq (1+\epsilon)^h\OPT[G,\varphi]$.
The Lipschitz constant is bounded by $31h(1+\epsilon)^h\epsilon^{-1}\log(g(t,|\varphi|)\epsilon^{-1})$ for some computable function $g$.
\end{lemma}
\begin{proof}
The approximability bound is obtained by repeatedly applying Lemmas~\ref{lem:paracom_approx_min}~and~\ref{lem:forget_approx_min}.
The Lipschitzness bound is obtained by repeatedly applying Lemmas~\ref{lem:paracom_lip_min}~and~\ref{lem:forget_lip_min} together with $\alpha\leq (1+\epsilon)^h$.
\end{proof}

Thus, we have the following.
\begin{proof}[Proof of Theorem~\ref{thm:alg_min}]
We prove the theorem by substituting $\epsilon$, $h$, and $t$ in Lemma~\ref{lem:combineall_min} with $\frac{\epsilon}{O(\log n)}$, $O(\log n)$, and $3t-2$, respectively.
Let $c\in \mathbb{R}_{>0}$ be the number such that $c\log n$ is an upper bound of $h$.
We substitute $\epsilon$ by $\frac{\epsilon}{2c\log n}$.
Then, the approximation ratio is bounded as
\begin{align*}
    \left(1+\frac{\epsilon}{2c\log n}\right)^{c\log n}\leq \lim_{n\to \infty}\left(1+\frac{\epsilon}{2n}\right)^n\leq e^\frac{\epsilon}{2}\leq 1+\epsilon,
\end{align*}
where the last inequality is from $\epsilon\in (0,1]$.
The Lipschitz constant is bounded as
\begin{align*}
    &31c\log n\left(1+\frac{\epsilon}{2c\log n}\right)^{c\log n}\cdot 2c\log n\epsilon^{-1}\log\left(g(3t-2,|\varphi|)\cdot 2c\log n\epsilon^{-1}\right)\\
    &\leq O\left(\epsilon^{-1}\log^2 n\log(g(3t-2,|\varphi|)\epsilon^{-1}\log n)\right)
    \leq O\left((f(t,|\varphi|)+\log \epsilon^{-1} + \log \log n)\epsilon^{-1}\log^2 n\right),
\end{align*}
where $f$ is some computable function.
\end{proof}

\subsection{Clique-Width}

% \snote{wrote}

Here, we provide the minimization version of the framework from \Cref{sec:DP_clique}.
Like in \Cref{sec:DP_clique}, the discussion and proof are the same as in \Cref{sec:minimize_tree}.
Therefore, we only provide the statements of the following lemmas.

\begin{lemma}\label{lem:paracom_approx_min_clique}
Let $0 < \alpha \leq 1$ and assume 
\begin{align*}
    \E\left[w(\DP[G,\varphi^G_i])\right]&\leq \alpha \OPT[G,\varphi^G_i],\\
    \E\left[w(\DP[H,\varphi^H_i])\right]&\leq \alpha \OPT[H,\varphi^H_i]
\end{align*}
holds for all $i\in [p]$.
Then, we have $\E\left[w\left(\DP[G\oplus H,\varphi]\right)\right]\leq (1+\epsilon)\alpha\OPT[G\oplus H,\varphi]$.
\end{lemma}

\begin{lemma}\label{lem:paracom_lip_min_clique}
Let $0 < \alpha \leq 1$, $\beta\in \mathbb{R}_{\geq 0}$ and assume 
\begin{align*}
    \E\left[w(\DP[G\oplus H,\varphi])\right]&\leq \alpha \OPT_w[G\oplus H,\varphi],\\
    \E\left[(w+\delta \bm{1}_u)(\DP[G\oplus H,\varphi])\right]&\leq \alpha \OPT_{w+\delta\bm{1}_u}[G\oplus H,\varphi],\\
    \EM\left(\DP_{w}[G,\varphi^{G}_i],\DP_{w+\delta\bm{1}_{u}}[G,\varphi^{G}_i]\right)&\leq \beta
\end{align*}
holds for all $i\in [p]$.
Then, we have
\begin{align*}
    \EM\left(\DP_{w}[G\oplus H,\varphi],\DP_{w+\delta\bm{1}_{u}}[G\oplus H,\varphi]\right) \leq 30\alpha\epsilon^{-1}\log(2p\epsilon^{-1})\delta + \beta.
\end{align*}
\end{lemma}

The following lemma corresponds to Lemma~\ref{lem:combineall_clique}.
The proof is the same as Lemma~\ref{lem:combineall_min}.
\begin{lemma}\label{lem:combineall_min_clique}
Let $\epsilon \in \mathbb{R}_{>0}$ and $h$ be the height of $T$.
Our algorithm outputs a solution $X\in \sat(G,\varphi)$ such that $\E[w(X)]\leq (1+\epsilon)^h\OPT[G,\varphi]$.
The Lipschitz constant is bounded by $30(h+1)(1+\epsilon)^h\epsilon^{-1}\log(g(t,|\varphi|)\epsilon^{-1})$ for some computable function $g$.
\end{lemma}
Theorem~\ref{thm:alg_min_clique} is also proved by the same proof as Theorem~\ref{thm:alg_min}.
\end{document}